\documentclass[accepted=2022-08-07, onecolumn, letterpaper, 11pt]{quantumarticle}
\pdfoutput=1

\usepackage[pagebackref,colorlinks,citecolor=blue,linkcolor=blue,bookmarks=true]{hyperref}
\usepackage{xspace}
\DeclareMathAlphabet{\mathbbold}{U}{bbold}{m}{n}
\usepackage{amssymb,amsmath,amsthm,graphicx,bbm}
\usepackage{mathtools}
\usepackage[usenames,dvipsnames,svgnames,table]{xcolor}
\usepackage{thmtools, thm-restate}
\usepackage[margin=1.in]{geometry}
\usepackage{multirow,array}
\usepackage{colortbl}
\allowdisplaybreaks
\usepackage{orcidlink}
\usepackage{tcolorbox}
\usepackage{mathrsfs}
\usepackage{microtype}
\usepackage[utf8]{inputenc}
\usepackage[english]{babel}
\usepackage[T1]{fontenc}
\usepackage[colorinlistoftodos]{todonotes}
\hypersetup{
    pdftitle={}, 
    pdfauthor={}, 
    colorlinks=true, 
    linkcolor=quantumviolet, 
    citecolor=quantumviolet, 
    filecolor=quantumviolet, 
    urlcolor=quantumviolet
}
\usepackage{caption}

\renewcommand{\backref}[1]{}

\renewcommand{\backrefalt}[4]{%
\ifcase #1 %
\or
[p.\ #2]%
\else
[pp.\ #2]%
\fi}

\usepackage{cite}
\usepackage[capitalise,nameinlink]{cleveref}

\usepackage{authblk}
\usepackage{algorithm}
\usepackage[noend]{algpseudocode}
\usepackage{tikz,tikz-qtree}
\usepackage{wrapfig}
\usepackage{enumitem}
\usepackage{tikz}
\usepackage{doi}
\newcommand*\circled[1]{\tikz[baseline=(char.base)]{
            \node[shape=circle,draw,inner sep=2pt] (char) {#1};}}

\theoremstyle{plain}
\newtheorem{theorem}{Theorem}[section]
\newtheorem*{theorem*}{Theorem}
\newtheorem{proposition}[theorem]{Proposition}
\newtheorem{lemma}[theorem]{Lemma}
\newtheorem{claim}[theorem]{Claim}
\newtheorem{fact}[theorem]{Fact}

\newtheorem{corollary}[theorem]{Corollary}

\theoremstyle{definition}
\newtheorem{definition}[theorem]{Definition}
\theoremstyle{remark}
\newtheorem{remark}[theorem]{Remark}
\theoremstyle{plain}

\def\Complex{{\mathbb{C}}} 

\renewcommand{\Pr}{\mathop{\bf Pr\/}}
\newcommand{\E}{\mathop{\mathbb E\/}}

\def\Udag{\mathcal U} 

\newcommand{\eps}{\varepsilon}

\DeclarePairedDelimiter\floor{\lfloor}{\rfloor}

\renewcommand{\tilde}{\widetilde}
\newcommand{\ket}[1]{{\left\vert #1 \right\rangle}}
\newcommand{\bra}[1]{{\left\langle #1 \right\vert}}
\newcommand{\ketbra}[2]{\ket{#1}\!\bra{#2}}

\newcommand{\sumCube}[1]{\sum_{#1\in \{0,1\}^n}}
\newcommand{\diag}{\mathrm{diag}}

\newcommand{\B}{\{0,1\}}
\newcommand{\dd}{\mathrm{d}}
\newcommand{\ii}{\mathrm{i}}
\newcommand{\nn}{\nonumber\\}

\newcommand{\shadow}{\mathrm{shadow}}

\let\OldLambda\lambda
\let\lambda\relax
\DeclareMathOperator{\lambda}{\OldLambda}

\DeclareMathOperator{\tr}{tr}
\DeclareMathOperator{\me}{\mathrm{e}}
\newcommand{\norm}[1]{\left\lVert#1\right\rVert}
\DeclareMathOperator{\Tr}{tr}

\DeclareMathOperator{\Var}{\mathrm{Var}}
\DeclareMathOperator{\median}{\mathrm{median}}

\newcommand{\channel}{\mathcal{E}}
\newcommand{\ensemble}{\mathcal{U}}

\newcommand{\shadowchannel}[2]{\mathcal{M}_{#1,#2}}
\newcommand{\linearoperator}{\mathcal{L}}

\newcommand{\inverseshadowchannel}[2]{\mathcal{M}_{#1, #2}^{-1}}

\newcommand{\localclifford}{\mathcal{C}_{1}}

\usepackage{ bbold }

\title{Classical Shadows With Noise}
\author{Dax Enshan Koh}
\affiliation{Institute of High Performance Computing, Agency for Science, Technology and Research (A*STAR), 1 Fusionopolis Way, \#16-16 Connexis, Singapore 138632, Singapore}
\affiliation{Zapata Computing,~Inc., 100 Federal Street, 20th Floor, Boston, Massachusetts 02110, USA}
\email{dax\textunderscore koh@ihpc.a-star.edu.sg}
\orcid{0000-0002-8968-591X}
\author{Sabee Grewal}
\affiliation{Department of Computer Science, The University of Texas at Austin, Austin, TX 78712, USA}
\affiliation{Zapata Computing,~Inc., 100 Federal Street, 20th Floor, Boston, Massachusetts 02110, USA}
\email{sabee@cs.utexas.edu}
\orcid{0000-0002-8241-560X}
\date{}

\begin{document}
\maketitle

\begin{abstract}
The classical shadows protocol, recently introduced by Huang, Kueng, and Preskill [Nat.~Phys.~16, 1050 (2020)], is a quantum-classical protocol to estimate properties of an unknown quantum state. Unlike full quantum state tomography, the protocol can be implemented on near-term quantum hardware and requires few quantum measurements to make many predictions with a high success probability.

In this paper, we study the effects of noise on the classical shadows protocol. In particular, we consider the scenario in which the quantum circuits involved in the protocol are subject to various known noise channels and derive an analytical upper bound for the sample complexity in terms of a shadow seminorm for both local and global noise. Additionally, by modifying the classical post-processing step of the noiseless protocol, we define a new estimator that remains unbiased in the presence of noise. As applications, we show that our results can be used to prove rigorous sample complexity upper bounds in the cases of depolarizing noise and amplitude damping.
\end{abstract}

\tableofcontents

\newpage
\section{Introduction}
\label{sec:intro}

Estimating the expectation values of quantum observables with respect to preparable quantum states is an important subroutine in many NISQ\footnote{NISQ---coined by Preskill \cite{preskill2018quantum}---stands for \textit{noisy intermediate-scale quantum}.}-era quantum algorithms that are of potential
practical importance \cite{preskill2018quantum,bharti2022noisy}. 
These algorithms, which include variational quantum algorithms \cite{cerezo2021variational} like the variational quantum eigensolver (VQE) \cite{peruzzo2014variational} and the quantum approximate optimization algorithm (QAOA) \cite{farhi2014quantum}, promise wide-ranging applications in, \textit{inter alia}, quantum chemistry \cite{cao2019quantum}, quantum metrology \cite{giovannetti2006quantum}, and optimization \cite{moll2018quantum}.
However, estimation is often the major bottleneck in many of these applications, where the number of measurements required is often too large for the algorithms to achieve the desired accuracy on useful instances using near-term quantum hardware \cite{wecker2015progress, huggins2021efficient}. 
Thus, developing efficient estimation protocols that can be implemented on near-term quantum hardware is critical to developing applications for NISQ devices.

In a recent breakthrough, Huang, Kueng, and Preskill introduced the \textit{classical shadows protocol} \cite{huang2020predicting},
a protocol for estimating many properties of a quantum state with few quantum measurements. 
The classical shadows protocol is based on the following idea: instead of recovering a full classical description of a quantum state like in full quantum state tomography \cite{haah2017sample, o2016efficient}, the protocol aims to learn only a minimal classical sketch---the classical shadow\footnote{The term `shadow' comes from Aaronson's work on \textit{shadow tomography} \cite{aaronson2020shadow}.}---of the state, which can then later be used to predict functions of the state (e.g., expectation values of observables). 

Classical shadows requires minimal quantum resources, yet can efficiently perform useful estimation tasks, making it amenable for use in the NISQ era. For example, classical shadows can efficiently estimate the energy of local Hamiltonians, verify entanglement, and estimate the fidelity between an unknown quantum state and a known quantum pure state (see \cite{huang2020predicting} for more applications). 
Additionally, rigorous performance guarantees on the protocol---in the form of upper bounds on the required number of samples in terms of error and confidence parameters---have been proved.

An assumption made in the original work by Huang, Kueng, and Preskill (and in some subsequent works by others) is that the unitary operators involved in the protocol can be executed perfectly. 
In real-world experiments with actual quantum hardware, however, this assumption will almost never hold due to the effects of noise on the quantum systems involved. 
Hence, for an accurate description of how the classical shadows protocol will perform in practice, it is important to take into account the effects of noise. 
The main contribution of this work is theoretically addressing how noise affects classical shadows. 
We derive rigorous sample complexity upper bounds for the most general noise channel, assuming only that the noise is described by a completely positive and trace-preserving linear superoperator. We also show how our results specialize in specific examples, e.g., when the noise is local, or when the noise is described by a depolarizing channel or an amplitude damping channel. 

\subsection{Main Ideas}
\label{sec:main_ideas}

\subsubsection{Review of Classical Shadows}
\label{sec:classical_shadows_review}

Classical shadows require the ability to perform computational basis measurements and apply a collection of unitary transformations, called the \textit{unitary ensemble}. 
The choice of unitary ensemble affects both the time complexity and the number of measurements needed (the sample complexity) for the protocol to succeed with small error.
For classical shadows to be time efficient, the unitary ensemble must be efficiently classically simulable on computational basis states (which is why considerable focus is given to the Clifford group in this work and in the original work \cite{huang2020predicting}).

Consider the following random process:
sample (with respect to some fixed probability distribution) a unitary transformation $U$ from the unitary ensemble $\mathcal U$. Apply $U$ to a quantum state $\rho$, 
and measure the resulting state $U \rho U^\dagger$ in the computational basis to get outcome state $\lvert b \rangle\!\langle b\rvert$. Finally,  classically simulate  
$\lvert b \rangle\!\langle b \rvert \mapsto U^\dagger \lvert b \rangle\!\langle b \rvert U$.
For any unitary ensemble, this process is a quantum channel in expectation, which we call the \textit{noiseless shadow channel}:
\begin{align}
 \label{eq:noiseless_shadow_channel}
    \mathcal M: \rho \mapsto \E_{U \sim \ensemble} \sum_{b \in \{0,1\}^n} \bra{b} U \rho U^\dagger \ket{b} U^\dagger \ket{b}\!\bra{b} U,
\end{align}
where $\E$ denotes the expectation value. A sufficient condition for the noiseless shadow channel to be invertible is that the unitary ensemble is tomographically complete \cite{huang2020predicting}.
\begin{definition}[tomographically complete]
A unitary ensemble $\ensemble$ is \textit{tomographically complete} if
for each pair of quantum states $\sigma \neq \rho$,
there exists a $U \!\in \ensemble$ and $b \in \{0,1\}^n$ such that $\bra{b}U \sigma U^\dagger \ket{b} \neq \bra{b}U \rho U^\dagger \ket{b}$.
\end{definition}

The classical shadows protocol works as follows: run the random process above to produce
a classical description of $U^\dagger \lvert b \rangle\!\langle b \rvert U$
(which involves quantum and classical computation) and then apply the inverse shadow channel $\mathcal{M}^{-1}$ (which involves only classical computation). 
The output $\hat{\rho} \overset{\mathrm{def}}{=} \mathcal M^{-1}( U^\dagger \lvert b \rangle\!\langle b \rvert U)$ is called the \textit{noiseless classical shadow}, which is an unbiased estimator of $\rho$, i.e., $\E[\hat{\rho}] = \rho$.  
Repeat this process to produce many classical shadows, a classical data set that can be used to estimate linear functions of the unknown state $\rho$. 

To estimate a linear function $\tr(O\rho)$, one must classically compute $\tr(O \hat{\rho})$ (an unbiased estimator of $\tr(O\rho)$) for each classical shadow from which a median-of-means estimator is constructed (see \cref{algo:median_of_means} for details). 
The power of the median-of-means estimator is captured in the following concentration inequality.
\begin{fact}[Jerrum et al.~\cite{jerrum1986random}]\label{fact:median-of-means}
Let $X$ be a random variable with variance $\sigma^2$. Then $K$ independent sample means of size $N = \frac{34 \sigma^2}{\eps^2}$ suffice to construct a median-of-means estimator $\hat{\mu}(N,K)$ that obeys 
\begin{align}
\Pr[ \lvert \hat{\mu}(N, K) - \E[X] \rvert \geq \eps ] \leq 2 \me^{-K/2}, \quad \forall\, \eps > 0.
\end{align}
\end{fact}

At this point, we have an unbiased estimator of $\tr(O\rho)$ which has nice concentration properties. However, the sample complexity depends on the variance $\Var[\tr(O\hat\rho)]$ (a function of the input state). 
To prove \textit{a priori} bounds on the sample complexity of classical shadows (i.e., bounds that do not depend on the input state), the authors introduce the \textit{shadow norm} $\norm{\cdot}_\mathrm{shadow}$, whose square is always an upper bound on the variance of $\tr(O \hat{\rho})$ (i.e.,
    $\Var[\tr(O\hat{\rho})] \leq \norm{O}_\shadow^2$).
Combining this upper bound with \cref{fact:median-of-means} yields the main result of \cite{huang2020predicting}:
\begin{theorem}[Informal version of Theorem 1 in Huang, Kueng, and Preskill~\cite{huang2020predicting}] \label{thm:informal_version}
Classical shadows of size $N$ suffice to estimate $M$ arbitrary linear functions $\tr(O_1 \rho), \ldots, \tr(O_M \rho)$ up to additive error $\eps$ given that 

\[
N \in O \left(\frac{\log(M)}{\eps^2} \max_{1 \leq i \leq M} \norm{O_i}^2_{\mathrm{shadow}} \right).
\]
\end{theorem}

The definition of the shadow norm depends on the unitary ensemble used to create the classical shadows. 
As examples, Huang, Kueng, and Preskill prove that if the unitary ensemble is the Clifford group, then 
$\norm{O}_\shadow^2 \leq 3\tr(O^2)$. In this case, the sample complexity is 

\begin{align}\label{eq:noiseless-sample-complexity-clifford}
N \in O \left(\frac{\log(M)}{\eps^2} \max_{1 \leq i \leq M} \tr(O_i^2)\right).
\end{align}

They also prove that if the unitary ensemble is the $n$-fold tensor product of the single-qubit Clifford group and the observable is a Pauli operator $P = P_1 \otimes \cdots \otimes P_n$, then 
$\norm{P}_\shadow^2 = 3^{\mathrm{wt}(P)}$, where $\mathrm{wt}(P) = \lvert \{ i : P_i \neq \mathbb I \}\rvert.$
In this case, the sample complexity is 

\begin{align}\label{eq:noiseless-sample-complexity-local-clifford}
N \in O \left(\frac{\log(M)}{\eps^2} \max_{1 \leq i \leq M} 3^{\mathrm{wt}(P_i)}\right).
\end{align}

To prove these bounds, Huang, Kueng, and Preskill use the fact that the Clifford group is a $3$-design (\cref{def:t-design}) \cite{zhu2016clifford, webb2016clifford}. 
They first express the shadow norm in terms of expectation values taken over the Clifford group, before replacing these expectations by integrals over the Haar measure by using the fact that the Clifford group forms a 3-design. Roughly speaking, this means that the uniform distribution over the Clifford group can duplicate properties of the probability distribution over the Haar measure for polynomials of degree not more than 3. These integrals have simple closed-form expressions, which can then be shown to be bounded by the expressions found in \cref{eq:noiseless-sample-complexity-clifford} and \cref{eq:noiseless-sample-complexity-local-clifford}.

\subsubsection{Our Contributions}
\label{sec:our_contributions}

In this paper, we consider the scenario in which the unitary operators involved in the classical shadows protocol are subject to noise. 
Specifically, we assume that an error channel $\cal E$ acts after the unitary operation is performed in the classical shadows protocol.\footnote{\label{footnote:GTM_main}This assumption on the noise model is sometimes referred to as the GTM noise assumption, where GTM stands for \textit{gate-independent}, \textit{time-stationary}, and \textit{Markovian} \cite{chen2020robust,flammia2020efficient,chen2022quantum}.
See \cref{sec:conclusion} for remarks on the scope and limitations of this assumption.}

The randomized measurement process remains the same as the noiseless protocol with the caveat that the system is subject to noise and continues to describe a quantum channel in expectation that we call 
the \textit{shadow channel with noise} $\cal E$:
\[\shadowchannel{\ensemble}{\channel}(\rho) = \E_{U \sim \ensemble} \sum_{b \in \{0,1\}^n} \bra{b} \channel(U \rho U^\dagger) \ket{b} U^\dagger \ket{b}\!\bra{b} U.
\] 
Note that since the shadow channel depends on noise, so does its inverse 
$\inverseshadowchannel{\ensemble}{\channel}$ (assuming it exists\footnote{We prove sufficient conditions for the invertibility of the shadow channel in the noisy setting (see \cref{claim:invertibility_global_clifford}).
\label{footenote:invertibility}
}). Therefore, while the classical shadows protocol remains similar 
 --- produce a classical description of $U^\dagger \ket{b}\!\bra{b}U$ and apply the inverse shadow channel ---
there is a necessary algorithmic change to the protocol when noise is present. Namely, in order for the classical shadow to remain an unbiased estimator, the classical post-processing step (i.e., when the inverse shadow channel is applied) must be modified to account for the noise. 
(We show this formally in \cref{sec:incorporatingNoise}.)

Noise also affects the sample complexity of the classical shadows shadows protocol. 
With noise present in the system, we prove the following sample complexity bounds, which generalize the main results of \cite{huang2020predicting}.
The key high-level takeaway is that the number of samples increases by only polynomial factors, suggesting that, even with noise, classical shadows can be run efficiently.
The bounds below are expressed in terms of the completely dephasing channel $\diag:\mathcal L(\mathbb C^d)\rightarrow \mathcal L(\mathbb C^d)$, which sends all the non-diagonal entries of its input to zero: $\diag(A) = \sum_{i=1}^d \ket i \!\bra i A \ket i \! \bra i$.

\begin{theorem}[Informal version of \cref{cor:sampleComplexityGlobal}] \label{thm:informal1}
Classical shadows of size $N$ suffice to estimate $M$ arbitrary linear functions $tr(O_1 \rho), \ldots, \tr(O_M \rho)$ to additive error $\eps$ 
when the quantum circuits used in the protocol are subject to the error channel $\channel$ given that  

\[
N \in O \left(\frac{2^{2n}\log(M)}{\beta^2\eps^2} \max_{1 \leq i \leq M} \tr(O_i^2)\right),
\]
where $\beta = \tr(\channel \circ \mathrm{diag})$.
\label{thm:informal_version_bound}
\end{theorem}

Here, $\beta$ is the trace (which we define in \cref{sec:prelim}) of the quantum channel $\mathcal E \circ \diag$, whose explicit form is given by $\beta = \sum_{i=1}^d \bra i \mathcal E (\ketbra{i}{i}) \ket i$; but, roughly speaking, one can think of $\beta$ as the ``severity of the noise'' on the quantum device. For instance, if we choose $\channel$ to be the identity channel (i.e., we model our device as noiseless), then $\beta^2 = 2^{2n}$, which recovers the noiseless sample complexity in \cref{eq:noiseless-sample-complexity-clifford}. 
Note that $\beta$ cannot be $0$, and so the upper bound in \cref{thm:informal_version_bound} is finite. We discuss this when we prove sufficient conditions for the invertibility of the shadow channel (\cref{subsec:global-derivation-of-shadow-channel}).

We also generalize the sample complexity bounds when the observables of interest are all Pauli operators.

\begin{theorem}[Informal version of \cref{corollary:product_clifford_sample_complexity}]\label{thm:informal-product-clifford}
Let $\{P_i\}_{i=1}^M$ be a collection of $M$ Pauli operators. Classical shadows of size $N$ suffice to estimate linear functions 
$\tr(P_1 \rho), \ldots, \tr(P_M \rho)$ to additive error $\eps$ 
when the quantum circuits used in the protocol are 
subject to quantum channel $\channel^{\otimes n}$ given that 
\[
N \in O \left(\frac{\log(M)}{\eps^2} \max_{1 \leq i \leq M} \bigg( \frac{3}{\beta ^2}\bigg)^{\mathrm{wt}(P_i)}\right), 
\]
where $\beta = \tr(\channel \circ \mathrm{diag})$ and $\mathrm{wt}(P) = \lvert \{ i : P_i \neq \mathbb I \}\rvert.$
\end{theorem}

It is easy to verify that if we choose $\channel$ to be the identity channel, then we recover the noiseless bound (\cref{eq:noiseless-sample-complexity-local-clifford}) proved in \cite{huang2020predicting}. It is important to note that, for this result, we assume that noise can be modelled on the device as a tensor product of single-qubit quantum channels. For simplicity, we have assumed that these single-qubit channels are identical (i.e.~each qubit is subject to the same noise model); we note, however, that it is straightforward to generalize this to the case where each single-qubit noise channel is different.

In addition to the sample complexity bounds above, we prove several new results. Among these are:
\begin{itemize}
    \item Tensor product noise cannot affect nice factorization properties of classical shadows with tensor product structure (\cref{subsection:product_ensemble_noise}).
    \item We prove a simple sufficient condition for the noisy shadow channel to be invertible (\cref{subsec:global-derivation-of-shadow-channel}). 
    \item If the unitary ensemble is a $2$-design, then the shadow channel is a depolarizing channel (even in the presence of a general quantum channel). 
    (\cref{subsec:global-derivation-of-shadow-channel}). 
    \item We prove nontrivial generalizations of shadow norm upper bounds presented in \cite{huang2020predicting} (\cref{subsec:global-derivation-of-shadow-norm} and \cref{subsection:product_clifford_shadow_norm}). For general noise models, the shadow norm ceases to be a norm. It, however, retains the properties of a seminorm. To this end, we shall refer to the `generalized shadow norm' as the \textit{shadow seminorm}.
    \item As applications, we show that our results can be used to prove rigorous sample complexity upper bounds in the cases of depolarizing noise and amplitude damping (\cref{subsec:global-examples} and \cref{subsec:local-examples}). We consider these noise models as they are
    good approximate models for quantum noise occurring in real quantum systems \cite{nielsen2010quantum}.
\end{itemize}

\subsection{Related Work}
\label{sec:related_work}

\subsubsection{Property Estimation and Quantum Tomography}
\label{sec:propertyEstimationLiterature}
There have been numerous works in the literature that can be cast as algorithms for estimating properties of quantum states. These include the following:

\begin{itemize}
\item General algorithms, such as \textit{quantum state tomography}, where the goal is to recover a classical description of an unknown quantum state $\rho$, given copies of $\rho$ \cite{hradil1997quantum, paris2004quantum,kohout2010optimal, banaszek2013focus, gross2010quantum}. Amongst these algorithms are sample-optimal protocols that use an asymptotically optimal number of samples but which require entangled measurements that act simultaneously on all the samples \cite{o2016efficient, haah2017sample}, and more experimentally friendly protocols that require only single-sample measurements \cite{flammia2012quantum, sugiyama2013precision, kueng2016low, kueng2017low, gu2020fast}.

\item \textit{Matrix product state tomography}, where it is assumed that the unknown quantum state is well-approximated by a matrix product state with low bond dimensions \cite{cramer2010efficient,lanyon2017efficient}.

\item \textit{Multi-scale entanglement renormalization ansatz (MERA) tomography}, for which a method for reconstructing multi-scale entangled states using a small number of efficiently implementable measurements and fast post-processing was developed \cite{cardinal2012practical}.

\item \textit{Neural network tomography}, which trains a classical deep neural network to represent quantum systems \cite{carrasquilla2019reconstructing,gao2017efficient}.

\item \textit{Overlapping quantum tomography}, which uses single-qubit measurements performed in parallel and the theory of perfect hash families to reconstruct $k$-qubit reduced density matrices of an $n$-qubit state with at most $e^{O(k)} \log^2(n)$ rounds of parallel measurements \cite{cotler2020quantum}. 

\item \textit{Shadow tomography} \cite{aaronson2020shadow, aaronson2019gentle, buadescu2021improved}, where the goal is to estimate $\tr(O_1\rho), \ldots, \tr(O_M\rho)$ to $\pm \eps$ accuracy, given a list of observables $O_1, \ldots, O_M$ and copies $\rho$. Classical shadows can be viewed as an efficient algorithm for shadow tomography in the special case that the shadow norm of the observables is small (e.g., observables with low rank). 
\end{itemize}

\subsubsection{Solving the Measurement Problem}
\label{sec:solvingMeasurementProblem}

There have been a number of results which focus on
reducing the number of measurements required in near-term quantum algorithms (i.e., solving the so-called ``measurement problem''). 
Recently, several methods have been proposed, for example, Pauli grouping \cite{kandala2017hardware, verteletskyi2020measurement}, unitary partitioning \cite{izmaylov2019unitary,zhao2020measurement}, engineered likelihood functions \cite{wang2021minimizing, koh2022foundations}, and deep learning models \cite{carrasquilla2019reconstructing}.
See Section 3 of \cite{huang2020predicting} for details on how classical shadows compares with other methods.
See also \cite{bharti2022noisy, wang2021minimizing,gonthier2020identifying} for more details on the ``measurement problem'' in near-term quantum algorithms. 

\subsubsection{Classical Shadows}
\label{sec:classical_shadows}

A number of works based on classical shadows have appeared since its introduction in \cite{huang2020predicting}. These include a generalization of classical shadows to the fermionic setting \cite{zhao2021fermionic,wan2022matchgate,gorman2022fermionic} and the use of classical shadows to estimate expectation values of molecular Hamiltonians \cite{hadfield2022measurements} and to detect bipartite entanglement in a many-body mixed state by estimating moments of the partially transposed density matrix \cite{elben2020mixed}. In addition, the first experimental implementation of classical shadows was carried out by Struchalin et al.~in a quantum optical experiment with high-dimensional spatial states of photons \cite{struchalin2021experimental}.

During the final stages of preparing version 1 \cite{koh2020classical} of our manuscript, we became aware of recent independent work by Chen, Yu, Zeng, and Flammia \cite{chen2020robust}, who also study ways to counteract noise in the classical shadows protocol. A key difference between their work and ours is that they do not assume that the noise model is known beforehand. Due to this, their strategy involves first learning the noise as a simple stochastic model before compensating for these errors using robust classical post-processing. In our manuscript, we have assumed that the user has modelled the noise on the device before implementing our protocol. This noise characterization can be carried out using efficient learning methods such as \cite{harper2020efficient,flammia2020efficient}.

Subsequent to version 1 \cite{koh2020classical} of our manuscript, several new extensions and applications of classical shadows have been developed \cite{li2021vsql, lukens2021bayesian, garcia2021quantum, hu2022hamiltonian, neven2021symmetry, huang2021efficient, acharya2021informationally, hillmich2021decision, hadfield2021adaptive, wu2021overlapped, rath2021quantum, zhang2021experimental, huang2021provably, huggins2022unbiasing, hu2021classical, flammia2021averaged, levy2021classical, kunjummen2021shadow, helsen2021estimating, chen2021exponential, notarnicola2021randomized, sack2022avoiding, bu2022classical, mcginley2022quantifying,liu2022experimental, lukens2021classical, lowe2021learning, huang2022learning,hu2022logical, seif2022shadow,elben2022randomized,boyd2022training,nguyen2022optimising,coopmans2022predicting,shivam2022classical,mcnulty2022estimating,wan2022matchgate,gorman2022fermionic}. Amongst these are extensions of the classical shadows framework to quantum channels \cite{levy2021classical, kunjummen2021shadow} and to more general ensembles, like locally scrambled unitary ensembles \cite{hu2021classical} and Pauli-invariant unitary ensembles \cite{bu2022classical}. Additional applications of classical shadows include avoiding barren plateaus in variational quantum algorithms \cite{sack2022avoiding}, quantifying information scrambling \cite{mcginley2022quantifying}, and estimating gate set properties \cite{helsen2021estimating}.

\subsubsection{Quantum Error Mitigation}
\label{sec:quantum_error_mitigation}

The last few years have seen the invention of several quantum error mitigation techniques \cite{endo2021hybrid} to suppress errors in NISQ devices, which are prone to errors but yet are not large enough for quantum error correction \cite{fowler2012surface, campbell2017roads, preskill2018quantum} to be implemented. Among these techniques are extrapolation methods \cite{li2017efficient,temme2017error,tiron2020digital} (e.g., Richardson extrapolation and exponential extrapolation), Clifford data regression \cite{czarnik2021error}, quantum subspace expansion \cite{mcclean2017hybrid}, and probabilistic error cancellation (also known as the quasi-probability method) \cite{temme2017error,endo2018practical}. Like classical shadows, these techniques involve repeated measurements and classical post-processing to obtain an estimator of the desired result. Unlike these techniques, though, our noisy classical shadows protocol incorporates error mitigation directly into the classical post-processing step, without requiring any additional quantum resources in the measurement process.

Some additional comparisons may be drawn between our noisy classical shadows protocol and the quasi-probability method, first proposed by Temme et al.\ for special channels \cite{temme2017error} and then extended by Endo et al.\ to practical Markovian noise \cite{endo2018practical}. The central idea behind the quasi-probability method is that for any (invertible) noise channel, its effects can be reversed by probabilistically implementing its inverse, by using the fact that while the inverse of a quantum channel may not be a quantum channel (and hence cannot be implemented physically by applying unitary operations to quantum states), it may be written as a linear combination of quantum channels (called basis operations). Like the quasi-probability method, our noisy classical shadows protocol reverses the effects of noise by effectively implementing the inverse of the noise channel (as part of implementing the inverse of the shadow channel). Unlike the quasi-probability method where the inverse of the noise channel is applied only probabilistically, our noisy classical shadows protocol applies the inverse of the noise channel deterministically. This is possible since the inverse is applied not as a physical operation on a quantum state, but as a mathematical operation on a classical description of a quantum state.

\section{Mathematical Preliminaries}
\label{sec:prelim}

Throughout this paper, we denote the set of linear operators on a vector space $V$ by $\mathcal L(V)$. The sets of Hermitian operators, unitary operators, and density operators on $\mathbb C^d$ 
are denoted by $\mathbb H_d$, $\mathbb U_d$, and $\mathbb D_d$ respectively. We denote the Haar measure on the $d$-dimensional unitary group by $\eta$.

For a linear operator $A$, the \textit{spectral norm} of $A$ is defined as 
\begin{align}
    \norm{A}_\mathrm{sp} = \max_{ x \in \mathbb C^n, \norm{x} = 1} \norm{Ax}.
\end{align}
When $A$ is Hermitian, the spectral norm may be written as 
\begin{align}
    \norm{A}_\mathrm{sp} = \max_{\sigma \in \mathbb D_d} \left|\tr(\sigma A)\right|.
\end{align}

The set of positive integers is denoted by $\mathbb Z^+$. The set of integers from 1 to $d$ is denoted as $[d] = \{1,2, \ldots, d\}$.
The Kronecker delta is denoted by $\delta_{xy}$. We will also use the following generalization of the Kronecker delta:
\begin{align}
    \delta_{x_1 x_2 \ldots x_n} = \begin{cases}
    1 & \text{if } x_1 = x_2 = \cdots = x_n
    \\
    0 & \mbox{otherwise.}
    \end{cases}
\end{align}

\subsection{Linear Superoperators and Quantum Channels}
\label{sec:linear_operators_and_quantum_channels}
We briefly review some properties of linear superoperators and quantum channels that will be used in this paper. For a more comprehensive introduction to quantum channels, see \cite{watrous2018theory,nielsen2010quantum}.

Let $\mathcal E:\mathcal L(\mathbb C^d)\rightarrow \mathcal L(\mathbb C^d)$ be a linear superoperator. We say that $\mathcal E$ is a \textit{quantum channel} if it is both completely positive and trace-preserving. We say that $\mathcal E$ is \textit{unital} if the identity operator is a fixed point of $\mathcal E$, i.e.~$\mathcal E(I) = I$. Every linear superoperator $\mathcal E$ admits a \textit{Kraus representation}:
\begin{align}
\label{eq:KrausRep}
    \mathcal E(A) = \sum_a J_a A K_a^\dag.
\end{align}
In the special case when $\mathcal E$ is also a quantum channel, $\mathcal E$ can be written as
\begin{align}
\label{eq:KrausRepChannel}
    \mathcal E(A) = \sum_a K_a A K_a^\dag,
\end{align}
where 
\begin{align}
    \sum_a K_a^\dag K_a = I.
\end{align}

The vector space of linear operators $\mathcal L(\mathbb C^d)$ is equipped with the Hilbert-Schmidt inner product $\langle A,B \rangle = \tr(A^\dag B)$, which is the default inner product on $\mathcal L(\mathbb C^d)$ that we will use in the rest of the paper. We denote the superoperator adjoint of a superoperator $A$ as $A^*$, and reserve the dagger $()^\dagger$ for the operator adjoint: $()^\dagger = \overline{()^T}$, where $()^T$ and $\overline{()}$ denote the operator transpose and complex conjugate with respect to the computational basis.

Next, we define the quantum channels we  consider in this work. 

\begin{definition}[completely dephasing channel]
$A \in \mathcal{L}(\mathcal C^d)$. $\diag:\mathcal L(\mathbb C^d)\rightarrow \mathcal L(\mathbb C^d)$ denotes the \textit{completely dephasing channel}:
\begin{align}
    \diag(A) = \sum_{i=1}^d \ket i \bra i A \ket i \bra i.
\end{align}
\end{definition}

\begin{definition}[depolarizing channel]\label{def:depolarizing-channel}
$A \in \mathcal{L}(\mathbb{C}^d)$. The qudit \textit{depolarizing channel} with depolarizing parameter $f\in \mathbb R$ is defined by
\begin{align}
    \mathcal D_{n,f}(A) =
    fA + (1-f)\tr(A)\frac{\mathbb I}{d}.
\end{align}
\end{definition}

In the definition above, we have allowed the depolarizing parameter to take any arbitrary value $f\in \mathbb R$. We note here however that
it is typical to restrict the depolarizing parameter to satisfy $f \in [0,1]$, especially when $\mathcal D_{n,f}$ is viewed as an error channel; for $f$ in this range, one could view $\mathcal D_{n,f}$ as a quantum channel that leaves density operators $\rho$ unchanged with probability $f$ and replaces $\rho$ with the maximally mixed state $\mathbb I/d$ with probability $1-f$. It is interesting to note, though, that it is not necessary for $f\in [0,1]$ in order for $\mathcal D_{n,f}$ to be a quantum channel (i.e.~a completely positive and trace preserving map). While $\mathcal D_{n,f}$ is trace-preserving for all $f\in \mathbb R$, it is easy to show that $\mathcal D_{n,f}$ is completely positive if and only if\footnote{This follows directly from the fact that the eigenvalues of the Choi matrix $J(D_{n,f})$ corresponding to the depolarizing channel \cite{watrous2018theory} are $\frac{1+f(d^2-1)}{d}$ and $\frac{1-f}{d}$.
}
\begin{align}
    -\frac{1}{d^2-1}\leq f \leq 1.
    \label{eq:when_is_depolarizing_channel_a_channel}
\end{align}
It then follows that $D_{n,f}$ is a quantum channel if and only if \cref{eq:when_is_depolarizing_channel_a_channel} is satisfied. This property will be used later in the discussion of \cref{claim:bounds_on_fE}.

\begin{definition}[amplitude damping channel]\label{def:amplitude-damping-channel}
The $n$-qubit \textit{amplitude damping channel} with parameter $p \in [0,1]$ is defined by
\begin{align}
    \mathrm{AD}_{n,p} = \mathrm{AD}_{1, p}^{\otimes n}
\end{align}
where 
\begin{align}
\mathrm{AD}_{1,p}: \begin{pmatrix}
\rho_{00} & \rho_{01} \\
\rho_{10} & \rho_{11} 
\end{pmatrix} 
\mapsto 
\begin{pmatrix}
\rho_{00} + (1-p)\rho_{11} & \sqrt{p} \rho_{01} \\
\sqrt{p} \rho_{10} & p \rho_{11}
\end{pmatrix}
\end{align}
is the amplitude damping channel on a single qubit, defined by the Kraus operators 
\begin{align}
K_{\mathrm{AD}0} = \begin{pmatrix}
1 & 0 \\
0 & \sqrt{p} 
\end{pmatrix}, \qquad 
K_{\mathrm{AD}1} = \begin{pmatrix}
0 & \sqrt{1 - p} \\
0 & 0
\end{pmatrix}.
\end{align}
\end{definition}

The trace of a linear superoperator $\mathcal E:\mathcal L(\mathbb C^d) \rightarrow \mathcal L(\mathbb C^d)$ is
\begin{align}
    \Tr(\mathcal E) = \sum_{ij} \langle E_{ij}, \mathcal E(E_{ij})\rangle,
\end{align}
where $E_{ij} = \ketbra ij$ and $\langle \cdot, \cdot \rangle$ is the Hilbert-Schmidt inner product.
Explicitly,
\begin{align}
    \tr(\mathcal E) &= \sum_{ij} \tr(E_{ij}^\dag \mathcal E(E_{ij})) \nn
    &= \sum_{ij} \bra i \mathcal E(\ketbra ij)\ket j.
\end{align}
It is straightforward to check that 
$\Tr(\mathcal E \circ \diag) = \sum_i \bra i \mathcal E(\ketbra ii) \ket i$, which is a quantity that appears often throughout this work.

\subsection{\texorpdfstring{$t$}{t}-Fold Twirls and \texorpdfstring{$t$}{t}-Designs}
\label{sec:twirlsDesigns}

$t$-designs are an important concept in quantum information processing with wide-ranging applications ranging from tensor networks \cite{nezami2020multipartitle} and quantum speedup \cite{brandao2013exponential,bouland2018complexity,mezher2019efficient}, to decoupling \cite{szehr2013decoupling} and quantum state encryption \cite{ambainis2009nonmalleable}. In this subsection, we shall review the definitions and some properties of $t$-fold twirls and $t$-designs that we will use in this paper. Throughout this subsection, we fix $d \in \mathbb Z_{\geq 2}$ to be an integer greater than or equal to 2. 

\begin{definition}[Twirl]
Let $\mathcal U \subseteq \mathbb{U}_d$ be a set of unitaries and let $t \in \mathbb{Z}^{+}$. The $t$-fold \textit{twirl} by $\mathcal U$ is the map $\Psi_{\mathcal U, t}: \mathcal{L}(\mathbb{C}^{d^t}) \rightarrow \mathcal{L}(\mathbb{C}^{d^t})$
defined by
\begin{align}
\Psi_{\mathcal U, t}(A)=\underset{U \sim \mathcal U}{\mathbb{E}} U^{\otimes t} A\left(U^\dag\right)^{\otimes t}.
\end{align}
\end{definition}

We denote the $t$-fold twirl by the Haar-random unitaries by $T_t^{(d)}: \mathcal L(\mathbb C^{d^t})
\rightarrow \mathcal L(\mathbb C^{d^t})$, i.e.~for all $A \in \mathcal L(\mathbb C^{d^t})$,
\begin{align}
\label{eq:t-foldTwirlHaar}
    T_t^{(d)}(A) = \Psi_{U(\mathbb C^d),t}(A) = \int \dd \eta(U) \ U^{\otimes t} A\left(U^\dag\right)^{\otimes t}.
\end{align}

When $t=2$, the Haar integral in \cref{eq:t-foldTwirlHaar} may be evaluated as
\begin{align} \label{eq:HaarIntegralt2}
    T_2^{(d)} (A) &= \frac{1}{d^2-1} \left[
    \tr(A)\left(I-\frac{W}{d}\right) +
    \tr(WA) \left(W-\frac{I}{d}\right)
    \right].
\end{align}
where 
\begin{align}
\label{eq:swapW}
    W = \sum_{i,j=1}^d\ket{ij}\bra{ji}
\end{align}
is the swap operator on $\mathbb C^d \otimes \mathbb C^d$. For a derivation of \cref{eq:HaarIntegralt2}, see \cite[Eq.~(7.179)]{watrous2018theory}.

The unitary ensembles considered in \cref{section:global_clifford} and \cref{sec:productClifford} of this paper are $t$-designs, collections of unitaries that reproduce the $t$-fold twirl by the Haar-random unitaries. 

\begin{definition}[$t$-design]\label{def:t-design}
Let $\mathcal U \subseteq \mathbb{U}_d$ be a finite set of unitaries, and let $t \in \mathbb Z^+$. We say that $\mathcal U$ is a $t$-\textit{design} if $\Psi_{\mathcal U,t} = T_t^{(d)}$, i.e.
\begin{align}
    \underset{U \sim \mathcal U}{\mathbb{E}} U^{\otimes t} A\left(U^\dag\right)^{\otimes t} = \int \dd \eta(U) \ U^{\otimes t} A\left(U^\dag\right)^{\otimes t}.
\end{align}
\end{definition}
Note that if $\mathcal U$ is a $t$-design, then it is also an $s$-design for all $s\leq t \in \mathbb Z^+$. An important example of a 3-design that fails to be a 4-design is the $n$-qubit Clifford group, denoted by $\mathcal C_n$ \cite{zhu2017multiqubit,webb2016clifford, zhu2016clifford}.

We now state a useful identity that we will use later in the paper.
\begin{lemma}
\label{lem:identity3designsum}
Let $\Udag$ be a qudit 3-design and $\mathcal E:\mathcal L(\mathbb C^d) \rightarrow \mathcal L(\mathbb C^d)$ be a linear superoperator. Let $A, B, C \in \mathcal L(\mathbb C^d)$ be linear operators. Then
\begin{align}
    & \E\limits_{U\sim \mathcal U} \sum_{b\in[d]} 
    \bra b \mathcal E(UAU^\dag) \ket b \bra b UBU^\dag \ket b \bra b UCU^\dag \ket b \nn
    &\quad = \frac{(1+d)\alpha-2\beta}{(d-1)d(d+1)(d+2)} 
    \left(
    \tr(A)\tr(BC) + \tr(A)\tr(B) \tr(C)
    \right) \nn
    &\qquad\quad+
    \frac{d\beta-\alpha}{(d-1)d(d+1)(d+2)}
    \left(\tr(AB) \tr(C)
    +
    \tr(AC) \tr(B)
    +
    \tr(ABC)
    +
    \tr(ACB)
    \right),
    \label{eq:twist_lemma}
\end{align}
where $\alpha = \tr(\mathcal E(I))$ and $\beta = \Tr(\mathcal E \circ \diag)$.
\end{lemma}
We present a proof of \cref{lem:identity3designsum} in \cref{app:identity3designsum}.

\section{Noisy Classical Shadows}
\label{sec:incorporatingNoise}

\subsection{Generating the Classical Shadow}

We study the setting where an error channel $\channel$ acts on the input state right after some $U$ from the unitary ensemble $\ensemble$ is applied.  We assume access to a noisy measurement primitive, slightly altered from \cite{huang2020predicting}. 

\begin{definition}[noisy measurement primitive]
\label{def:noisy_measurement_primitive}
We can apply a restricted set of unitary transformations $\rho \mapsto U \rho U^\dagger$, where $U$ is chosen uniformly at random from a unitary ensemble $\ensemble$. Subsequently, an error channel $\channel$ acts on the state $U \rho U^\dagger \mapsto \channel(U \rho U^\dagger)$. 
Finally, we can measure the state in the computational basis $\{\ket{b} : b \in \{0, 1\}^n\}$. 
\end{definition}

The randomized measurement procedure remains the same as \cite{huang2020predicting} with the caveat that the transformed state is subject to an error channel. The randomized measurement procedure is as follows.
Given copies of an input state $\rho$, perform the following on each copy: transform $\rho \mapsto \channel(U \rho U^\dagger)$, measure in the computational basis, and apply $U^\dagger$ to the post-measurement state. The output of this procedure is

\begin{align}
\label{eq:output_procedure}
U^\dagger \lvert \hat{b} \rangle\!\langle \hat{b} \rvert U \quad \text{with probability} \quad P_b(\hat{b}) \overset{\text{def}}{=} \langle \hat{b}\rvert \channel(U \rho U^\dagger) \lvert \hat{b} \rangle \quad \text{where} \quad \hat{b} \in \{0,1\}^n.
\end{align}

In expectation, this procedure describes a quantum channel.

\begin{definition}[shadow channel]
    Let $\shadowchannel{\ensemble}{\channel}: \linearoperator(\Complex^{2^n}) \rightarrow \linearoperator(\Complex^{2^n})$ be defined by 
    
    \begin{align}\label{eq:shadow_channel}
    \shadowchannel{\ensemble}{\channel}(\rho) \overset{\mathrm{def}}{=} \E_{U \sim \ensemble} \sum_{b \in \{0,1\}^n} \bra{b} \channel(U \rho U^\dagger) \ket{b} U^\dagger \ket{b}\!\bra{b} U = 
    \E_{\substack{U \sim \ensemble \\ \hat{b} \sim P_b }} U^\dagger |\hat b\rangle\!\langle\hat b| U.
    \end{align}
    We call $\shadowchannel{\ensemble}{\channel}$ the shadow channel with noise $\channel$.
\end{definition}

One can view $\shadowchannel{\ensemble}{\channel}$ as the expected output of the random measurement procedure described above. Note that when we take the channel $\channel$ to be the identity channel, we recover the \textit{noiseless shadow channel} $\mathcal{M}$ given by \cref{eq:noiseless_shadow_channel}. 

\begin{claim}
For all unitary ensembles $\ensemble$ and quantum channels $\channel$, $\mathcal M_{\ensemble, \channel}$ is a quantum channel. 
\end{claim}
\begin{proof}
$\mathcal M_{\ensemble, \channel} = \mathcal D \circ \mathcal C \circ \mathcal B \circ \mathcal A$ is a composition of quantum channels, where
$\mathcal A(X) = \E_{U \sim \ensemble} U X U^\dagger$ (mixed unitary channel), $\mathcal B(X) = \channel(X)$ (error channel), 
$\mathcal{C}(X) = \sum_b \ket{b}\!\bra{b}X \ket{b}\!\bra{b}$ (quantum-to-classical channel), and 
$\mathcal{D}(X) = U^\dagger X U$ (unitary channel).
\end{proof}

In the noiseless case, it is known that if the unitary ensemble is tomographically complete, then the shadow channel is invertible \cite{huang2020predicting}. 
Similarly, we prove sufficient conditions for the shadow channel to be invertible in the noisy case (see \cref{claim:invertibility_global_clifford}).
Assuming the shadow channel is invertible, we can define the classical shadow. 

\begin{definition}[classical shadow]
Assuming the shadow channel is invertible with inverse $\inverseshadowchannel{\ensemble}{\channel}$, define the \textit{classical shadow} $\hat{\rho}$ as 
\begin{align}\label{eq:classical_shadow}
\hat{\rho} = \hat{\rho}(\ensemble, \channel, \hat{U}, \hat{b}) \overset{\mathrm{def}}{=} 
\inverseshadowchannel{\ensemble}{\channel}(\hat{U}^\dagger \lvert \hat{b}\rangle\!\langle \hat{b} \rvert \hat{U}).
\end{align}
\end{definition}

The classical shadow is a random matrix with unit trace\footnote{
    $\mathrm{tr}(\hat{\rho}) 
    = \mathrm{tr}( \inverseshadowchannel{\ensemble}{\channel}(\hat{U}^\dagger \lvert \hat{b}\rangle\!\langle \hat{b} \rvert \hat{U}))
    = \mathrm{tr}( \hat{U}^\dagger \lvert \hat{b}\rangle\!\langle \hat{b} \rvert \hat{U})
    = 1.$ We use the fact that the inverse shadow channel is trace preserving because the shadow channel is trace preserving. 
\label{footnote:classical-shadow-trace-one}}
and reproduces $\rho$ in expectation: $\E[\hat{\rho}] = \rho$.\footnote{
    $\E [\hat{\rho}] 
    = \E[  \inverseshadowchannel{\ensemble}{\channel}(\hat{U}^\dagger \lvert \hat{b}\rangle\!\langle \hat{b} \rvert \hat{U})]
    = \inverseshadowchannel{\ensemble}{\channel}(\E [\hat{U}^\dagger \lvert \hat{b}\rangle\!\langle \hat{b} \rvert \hat{U}])  
    = \inverseshadowchannel{\ensemble}{\channel}(\shadowchannel{\ensemble}{\channel} (\rho)) = \rho. 
    $
} Repeating this process $N$ times produces a classical shadow with size $N$.

\begin{definition}[size-$N$ classical shadow]
The size-$N$ \textit{classical shadow} corresponding to pairs $(U_1, \hat{b}_1),  \ldots, (U_N, \hat{b}_N )$ is 

\begin{align}
\mathsf{S}(\rho; N) = \{\hat{\rho}_1, \ldots, \hat{\rho}_N \} \quad \text{where} \quad \hat{\rho}_i = \inverseshadowchannel{\ensemble}{\channel}(U_i^\dagger \lvert \hat{b}_i\rangle\!\langle\hat{b}_i\rvert U_i).
\end{align}
\end{definition}

\subsection{Noisy Classical Shadows Protocol}

The classical shadow is a classical dataset that can be used to predict many linear functions in the unknown state $\rho$. Recall that a linear function in a quantum state $\rho$ is a function of the form $\rho \mapsto \tr(O \rho)$ for some linear operator $O$.
It is easy to confirm that the random variable $\tr(O \hat{\rho})$ reproduces $\tr(O\rho)$ in expectation: $\E [\tr(O \hat{\rho})] = \tr(O\E[\hat{\rho}]) = \tr(O\rho)$. 
Therefore, we can use the classical shadow to produce unbiased estimates of $\tr(O_1 \rho), \ldots, \tr(O_M \rho)$ for any observables $O_1, \ldots, O_M$. 
We continue to use median-of-means estimation, as was done in the noiseless protocol: 

\begin{algorithm}[H]
\caption{Median-of-means estimation based on a classical shadow.}
\textbf{Input:} a list of observables $O_1, \ldots, O_M$, size-$L$ classical shadow $\mathsf{S}(\rho; L)$, $K \in \mathbb{Z}^+$.
\label{algo:median_of_means}
\begin{algorithmic}[1]
\State Set $\hat{\rho}_{(k)} = \frac{1}{\floor{L/K}} \sum_{i = (k - 1)\floor{L/K} + 1}^{k \floor{L/K}} \hat{\rho}_i, \,\, \text{for $k = 1, \ldots, K$}$
\State Output $\hat{o}_i(\floor{L/K}, K) = \median\big\{ \tr(O_i \hat{\rho}_{(1)}), \ldots, \tr(O_i \hat{\rho}_{(K)}) \big\}, \,\, \text{for $i = 1, \ldots, M$}$
\end{algorithmic}
\end{algorithm}

Using the concentration properties of median-of-means estimators (see \cref{fact:median-of-means}), we  understand how estimates $\tr(O\hat{\rho})$ concentrate around the true value as a function of $\Var[\tr(O \hat{\rho})]$. 
However, we are interested in bounds that are independent of the input state $\rho$, which motivates the following lemma.

\begin{lemma}\label{lemma:upper-bound-on-variance}
Let $\ensemble$ be a set of $n$-qubit unitary transformations and let $\channel$ be an $n$-qubit quantum channel. Assume that the shadow channel $\shadowchannel{\ensemble}{\channel}$ (\cref{eq:shadow_channel}) is invertible. Let $O \in \mathbb H_{2^n}$ and $\rho \in \mathbb{D}^{2^n}$ be an unknown $n$-qubit state. Let $\hat{o} = \tr(O \hat{\rho})$, where $\hat{\rho}$ is the classical shadow (\cref{eq:classical_shadow}). 
Then, 

\begin{align}
\underset{\substack{U \sim \ensemble \\ b \sim P_{b}}}{\Var}[\hat{o}] 
\leq \norm{O - \tr(O)\frac{\mathbb{I}}{2^n} }_{\shadow, \ensemble,\channel}^2 ,
\end{align}
where 
\begin{align}
\label{eq:noisy_shadow_norm}
\norm{O}_{\shadow, \ensemble, \channel} = \max_{\sigma \in \mathbb{D}_{2^n}} 
\sqrt{\underset{U \sim \ensemble}{\E} \sum_{b \in \{0,1\}^n} \langle b \rvert \channel(U \sigma U^\dagger) \lvert b \rangle\!\langle b \rvert U \mathcal{M}^{-1,\dagger}_{\ensemble, \channel}(O) U^\dagger \lvert b \rangle^2 }.
\end{align}
\end{lemma}

We call the function $\norm{\cdot}_{\shadow, \ensemble, \channel}$, which depends on only the unitary ensemble and the error channel, the \textit{shadow seminorm}. As we show in \cref{app:when-is-shadow-norm}, the shadow seminorm is indeed a \textit{seminorm}, i.e., it satisfies absolute homogeneity and the triangle inequality. However, unlike the noiseless case \cite{huang2020predicting}, the noisy shadow seminorm is not necessarily a norm: there exist noise channels $\mathcal E$ for which $\norm{\cdot}_{\shadow, \ensemble, \channel}$ fails to satisfy the point-separating property required of a norm. In \cref{app:when-is-shadow-norm}, we also explore the question about when the shadow seminorm is a norm. In particular, we prove that a sufficient condition for it to be a norm is that $\mathcal E$ satisfies the following property: for all $b\in \{0,1\}^n$, there exists a density operator $\sigma \in \mathbb D(\mathbb C^{2^n})$ such that $\bra b \mathcal E(\sigma)\ket b\neq 0$.

Again, the motivation for introducing the shadow seminorm is to get an upper bound on $\Var[\hat{o}]$ that does not depend on the unknown state $\rho$. The proof is a straightforward generalization of Lemma S1 in \cite{huang2020predicting}, which we defer to \cref{subsec:proof-of-lemma-upper-bound-on-variance}.

Thus far, we have shown that the noisy classical shadow is an unbiased estimator of linear functions in $\rho$ and have proved an upper bound on the variance of the estimator. 
This is enough to prove the following performance guarantee on the noisy classical shadows protocol.

\begin{theorem}\label{thm:general-theorem}
Fix an $n$-qubit unitary ensemble $\ensemble$, a collection of $n$-qubit observables $O_1, \ldots, O_M$, an $n$-qubit quantum channel $\channel$, 
and accuracy parameters $\eps, \delta \in [0,1]$.
Assume that $\shadowchannel{\ensemble}{\channel}$ (\cref{eq:shadow_channel}) is invertible.
Set 
\[
K = 2 \log \frac{2M}{\delta} 
\qquad 
\text{and}
\qquad 
N = \frac{34}{\eps^2} \underset{1\leq i \leq M}{\max} \norm{O_i - \frac{1}{2^n} \tr(O_i)\mathbb{I}}^2_{\shadow, \ensemble, \channel}.
\]
Then, a size-$(NK)$ classical shadow $\mathsf{S}(\rho; NK)$ is sufficient to estimate $\hat{o}_1, \ldots, \hat{o}_M$ with the following performance guarantee: 
\begin{align}
    \Pr\! \Big[ \lvert \hat{o}_i(N,K) - \tr(O_i\rho) \rvert \leq \eps \quad \forall i = 1, \ldots, M \Big] \geq 1 - \delta.
\end{align}
Hence, the sample complexity to estimate a collection of $M$ linear target functions $\tr(O_i \rho)$ within error $\eps$ and failure probability $\delta$ is 
\[
NK = O\bigg( \frac{\log(M/\delta)}{\eps^2} \max_{1 \leq i \leq M} \norm{O_i - \frac{1}{2^n} \tr(O_i)\mathbb{I}}^2_{\shadow, \ensemble, \channel}\bigg).
\]

\end{theorem}
\begin{proof}
Run \cref{algo:median_of_means} with $O_1, \ldots, O_M$, $\mathsf{S}(\rho; NK)$, $N$, and $K$ to obtain estimates
\begin{align}
\hat{o}_i(N, K) = \median\big\{ \tr(O_i \hat{\rho}_{(1)}, \ldots, \tr(O_i \hat{\rho}_{(K)}) \big\},
\qquad \text{for $i = 1, \ldots, M$.}
\end{align}
Then, 
\begin{align}
    \Pr\! \Big[ \lvert \hat{o}_i(N,K) - \tr(O_i\rho) \rvert \leq \eps \,\, \forall i = 1, \ldots, M \Big] &= 1 - \Pr\Big[ \exists i = 1, \ldots, M : \lvert \hat{o}_i(N,K) - \tr(O_i\rho) \rvert \leq \eps  \Big] \nn
    &\geq 1 - \sum_{i=1}^M \Pr\Big[ \lvert \hat{o}_i(N,K) - \tr(O_i\rho) \rvert > \eps  \Big]  \nn
    &\geq 1 - 2\me^{-K/2} \sum_{i=1}^M 1  \nn
    &= 1 - \delta.
\end{align}
The inequality on the second line follows from the union bound and the inequality on the third line follows from \cref{fact:median-of-means}.
\end{proof}

The noisy classical shadows protocol is summarized next.

\begin{tcolorbox}[colback=white, title=Summary: Classical Shadows With Noise, title filled]
\textbf{Hyperparameters} \\ 
Let $\ensemble$ be a set of $n$-qubit unitary transformations. Let $\channel$ be an $n$-qubit quantum channel. \\

\textbf{Definitions} \\ 
Let 
\[\shadowchannel{\ensemble}{\channel}: \rho \mapsto \underset{U \sim \ensemble}{\E} \sum_{b\in\{0,1\}^n} \langle b \rvert \channel(U \rho U^\dagger)\lvert b \rangle U^\dagger \lvert b \rangle\!\langle b \rvert U,\] and assume that $\shadowchannel{\ensemble}{\channel}$ is invertible. 
Let 
\[
\norm{O}_{\mathrm{shadow}, \ensemble, \channel} = \max_{\sigma \in \mathbb{D}_{2^n}} 
\sqrt{\underset{U \sim \ensemble}{\E} \sum_{b \in \{0,1\}^n} \langle b \rvert \channel(U \sigma U^\dagger) \lvert b \rangle\!\langle b \rvert U \mathcal{M}^{-1,\dagger}_{\ensemble, \channel}(O) U^\dagger \lvert b \rangle^2 }. 
\]

\begin{algorithm}[H]
\caption{Classical Shadows With Noise}
\textbf{Input} \\ 
 $\rho \in \mathbb{C}^{2^n}$ (an unknown $n$-qubit state, given as multiple copies of a black box). $\eps, \delta \in (0,1)$ (accuracy parameters). $O_1, \ldots, O_M$ (a list of observables).
\vspace{0.1in}

\textbf{Output} \\ 
Estimators $\hat{o}_1, \ldots, \hat{o}_M$ such that 

\[
    \Pr\! \Big[ \lvert \hat{o}_i - \tr(O_i\rho) \rvert \leq \eps \quad \forall i = 1, \ldots, M \Big] \geq 1 - \delta.
\]

\label{algo:full_protocol}
\begin{algorithmic}[1]
\Statex \textbf{Initialization}
\State Set $K = 2 \log \frac{2M}{\delta}$
\State \label{step:full_protocol_2} Set $N = \frac{34}{\eps^2} \underset{1\leq i \leq M}{\max} \norm{O_i - \frac{1}{2^n} \tr(O_i)\mathbb{I}}^2_{\mathrm{shadow},\ensemble, \channel}$
\Statex
\Statex \textbf{Classical shadow generation}
\For {$i=1,\ldots,NK$}
\State Randomly choose $\hat{U}_i \in \ensemble$
\State \label{step:full_protocol_5} Apply $\rho \mapsto \channel(U \rho U^\dagger)$ to (a fresh copy of) $\rho$ to get $\rho_1$
\State Perform a computational basis measurement on $\rho_1$ to get outcome $\hat{b}_i \in \{0,1\}^n$ 
\State Save a classical description of $\hat{U}_i^\dagger \lvert \hat{b}_i\rangle\!\langle \hat{b}_i \rvert \hat{U}_i$ in classical memory
\State \label{step:full_protocol_8} Apply $\inverseshadowchannel{\ensemble}{\channel}$ to $\hat{U}_i^\dagger \lvert \hat{b}_i\rangle\!\langle \hat{b}_i \rvert \hat{U}_i $ to get $\hat{\rho}_i = \inverseshadowchannel{\ensemble}{\channel}(\hat{U}_i^\dagger \lvert \hat{b}_i\rangle\!\langle \hat{b}_i \rvert \hat{U}_i)$
\EndFor
\State Set $\mathsf{S}(\rho; NK) = \{\hat{\rho}_1, \ldots, \hat{\rho}_{NK}\}$
\Statex
\Statex \textbf{Median-of-means estimation}
\State Set $\hat{\rho}_{(k)} = \frac{1}{N} \sum_{i = (k - 1)N + 1}^{k N} \hat{\rho}_i, \,\, \text{for $k = 1, \ldots, K$}$
\State Output $\hat{o}_i \overset{\mathrm{def}}{=} \hat{o}_i(N, K) = \median\big\{ \tr(O_i \hat{\rho}_{(1)}), \ldots, \tr(O_i \hat{\rho}_{(K)}) \big\}, \,\, \text{for $i = 1, \ldots, M$}$
\end{algorithmic}
\end{algorithm}
\end{tcolorbox}

\subsection{Product Ensembles with Product Noise}\label{subsection:product_ensemble_noise}

We conclude this section by showing some nice factorization properties for classical shadows when the unitary ensemble is a product ensemble and the quantum channel is a product channel. 

\begin{definition}[product channel]\label{def:product-channel}
An $n$-qubit \textit{product channel} is a quantum channel $\channel$ of the form $\channel = \channel_1 \otimes \ldots \otimes \channel_n: \mathcal{L}(\mathbb C^{2^n}) \rightarrow \mathcal{L}(\mathbb C^{2^n})$, where each $\channel_i: \mathcal{L}(\mathbb C^{2}) \rightarrow \mathcal{L}(\mathbb C^{2})$.
\end{definition}

\begin{definition}[product ensemble]\label{def:product-ensemble}
An $n$-qubit \textit{product ensemble} is a collection of unitary transformations of the form $\ensemble = \bigotimes_{i=1}^n \ensemble_i = \{U_1 \otimes \ldots \otimes U_n : U_1 \in \ensemble_1, \ldots , U_n \in \ensemble_n\}$.
\end{definition}

First, we show that the shadow channel, the inverse shadow channel, and the classical shadow all factorize into tensor products when the unitary ensemble is a product ensemble and the quantum channel is a product channel.

\begin{claim}\label{claim:product_shadow_channel_factorizes}
Let $\ensemble = \bigotimes_{i=1}^n \ensemble_i$ be a product ensemble, and let $\channel = \channel_1 \otimes \ldots \otimes \channel_n$ be a product channel. Then, the shadow channel factorizes as follows:
\[
\mathcal M_{\ensemble, \channel} = \bigotimes_{i=1}^n \mathcal M_{\ensemble_i, \channel_i}.
\]
Assume that $\mathcal M_{\ensemble_i, \channel_i}$ is invertible $\forall i \in \{1, \ldots, n\}$. Then, $\mathcal M_{\ensemble,\channel}$ is invertible and the inverse shadow channel factorizes as follows:  
\[
\mathcal M_{\ensemble, \channel}^{-1} = \bigotimes_{i=1}^n \mathcal M_{\ensemble_i, \channel_i}^{-1}.
\]
Finally, let $\hat{U} = \hat{U}_1 \otimes \ldots \otimes \hat{U}_n \in \ensemble$ and $\hat{b} = \hat{b}_1\ldots\hat{b}_n \in \{0,1\}^n$. Then, 
\begin{align}
    \hat{\rho}(\ensemble, \channel, \hat{U}, \hat{b}) = \bigotimes_{i=1}^n \hat{\rho}(\ensemble_i, \channel_i, \hat{U}_i, \hat{b}_1).
\end{align}
\end{claim}
\begin{proof}
The first part of the claim follows from two basic facts: 
Tensor products of quantum channels factorize when applied to elementary tensor products, and the $n$-th order tensor product is the linear hull of all elementary tensor products. 
The second part of the claim follows from the fact that $(A \otimes B)^{-1} = A^{-1} \otimes B^{-1}$.

The third part of the claim follows from the following chain of equalities:
\begin{align}
    \hat{\rho}(\ensemble, \channel, \hat{U}, \hat{b}) 
    = \mathcal M_{\ensemble, \channel}^{-1}(\hat{U}^\dagger \lvert \hat{b}\rangle\!\langle\hat{b}\rvert \hat{U}) = \bigotimes_{i=1}^n \mathcal M_{\ensemble_i, \channel_i}^{-1}(\hat{U}_i^\dagger \lvert \hat{b}_i\rangle\!\langle\hat{b}_i\rvert \hat{U}_i) = \bigotimes_{i=1}^n \hat{\rho}(\ensemble_i, \channel_i, \hat{U}_i, \hat{b}_1).
\end{align}
\end{proof}

\begin{claim}
Let $\ensemble = \bigotimes_{i=1}^n \ensemble_i$ be a product ensemble, and let $\channel = \channel_1 \otimes \ldots \otimes \channel_n$ be a product channel. 
Assume that $\mathcal M_{\ensemble_i, \channel_i}$ is invertible $\forall i \in \{1, \ldots, n\}$. 
Let $\hat{U} = \hat{U}_1 \otimes \ldots \otimes \hat{U}_n \in \ensemble$ and $\hat{b} = \hat{b}_1\ldots\hat{b}_n \in \{0,1\}^n$. Then, 
\begin{align}
    \hat{\rho}(\ensemble, \channel, \hat{U}, \hat{b}) = \bigotimes_{i=1}^n \hat{\rho}(\ensemble_i, \channel_i, \hat{U}_i, \hat{b}_1).
\end{align}
\end{claim}
\begin{proof}
By \cref{claim:product_shadow_channel_factorizes}, given a product ensemble and product channel, the shadow channel is $\mathcal{M}_{\ensemble, \channel} = \bigotimes_{i=1}^n \mathcal M_{\ensemble_i, \channel_i}$, and 
$\mathcal{M}_{\ensemble, \channel}^{-1} = \bigotimes_{i=1}^n \mathcal M_{\ensemble_i, \channel_i}^{-1}$. 
\begin{align}
    \hat{\rho}(\ensemble, \channel, \hat{U}, \hat{b}) 
    &= \mathcal M_{\ensemble, \channel}^{-1}(\hat{U}^\dagger \lvert \hat{b}\rangle\!\langle\hat{b}\rvert \hat{U})  \nn
    &= \bigotimes_{i=1}^n \mathcal M_{\ensemble_i, \channel_i}^{-1}(\hat{U}_i^\dagger \lvert \hat{b}_i\rangle\!\langle\hat{b}_i\rvert \hat{U}_i) \nn
    &= \bigotimes_{i=1}^n \hat{\rho}(\ensemble_i, \channel_i, \hat{U}_i, \hat{b}_1).
\end{align}
\end{proof}

We conclude this section with a nontrivial generalization of Proposition S2 in \cite{huang2020predicting}, which shows that tensor product noise cannot affect the nice factorization properties of classical shadows with tensor product structure.

\begin{lemma}\label{lemma:general_locality_respecting_shadow_norm}
Let $0 \leq k \leq n$. 
Let $O \in \mathbb H_2^{\otimes n}$ be an $n$-qubit operator that acts nontrivially as $\tilde{O} \in \mathbb H_2^{\otimes k}$ on $k$ qubits $i_1, \ldots, i_k$. 
Let $\ensemble = \bigotimes_{i=1}^n \ensemble_i$ be a product ensemble, and let $\channel = \channel_1 \otimes \ldots \otimes \channel_n$ be a product channel.
Assume that $\mathcal M_{\ensemble, \channel}$ is invertible. If $k=0$, then
$\norm{O}_{\mathrm{shadow},\ensemble, \channel}=1$. Otherwise, if $k\geq 1$, then
\begin{align}
    \norm{O}_{\mathrm{shadow},\ensemble, \channel} = \norm{\tilde{O}}_{\mathrm{shadow}, \mathcal U_{i_1}\otimes \cdots \otimes
    \mathcal U_{i_k},
    \mathcal E_{i_1} \otimes \cdots \otimes
    \mathcal E_{i_k}.
    }
    \label{eq:general_locality_respecting_shadow_norm}
\end{align}
\end{lemma}

\begin{proof}
By \cref{claim:product_shadow_channel_factorizes}, given a product ensemble and product channel, the shadow channel and inverse shadow channel factorize.  
It follows that that $\mathcal M_{\ensemble, \channel}^{-1, \dagger} = \bigotimes_{i=1}^n \mathcal{M}_{\ensemble_i, \channel_i}^{-1, \dagger}$. Additionally, the trace preserving property of $\mathcal M_{\ensemble_i, \channel_i}$ implies that $\mathcal M_{\ensemble_i, \channel_i}^{-1}$ is trace preserving. The complex conjugate of a trace-preserving quantum channel is unital (see \cite{watrous2018theory}, Theorem 2.26). Therefore, $\mathcal M_{\ensemble_i, \channel_i}^{-1, \dagger}$ is unital.

If $k = 0$, then $O = \mathbb I$. It follows from \cref{footnote:classical-shadow-trace-one} that  $\norm{\mathbb I}^2_{\mathrm{shadow}, \ensemble, \channel} = 1$.
Without loss of generality, take $O = \tilde{O} \otimes \mathbb I^{\otimes (n - k)}$.

\begin{align}
    \norm{O}_{\mathrm{shadow}, \ensemble, \channel}^2 
    &= \norm{\tilde{O} \otimes \mathbb I^{\otimes (n-k)}}_{\mathrm{shadow}, \ensemble, \channel}^2\nn
    &= \max_{\sigma \in \mathbb{D}_{2^n}} \underset{U \sim \ensemble}{\E} \sum_{b \in \{0,1\}^n} \langle b \rvert \mathbb \channel(U \sigma U^\dagger) \lvert b \rangle\!\langle b \rvert U \mathcal{M}_{\ensemble, \channel}^{-1, \dagger}(\tilde{O} \otimes \mathbb I^{\otimes (n -k)}) U^\dagger \lvert b\rangle^2 \nn
    &= \max_{\sigma \in \mathbb{D}_{2^n}} \underset{U \sim \ensemble}{\E} \sum_{b \in \{0,1\}^n} \langle b \rvert \mathbb \channel(U \sigma U^\dagger) \lvert b \rangle\!\langle b \rvert U \left( \bigotimes_{i = 1}^k \mathcal{M}_{\ensemble_i, \channel_i}^{-1, \dagger}(\tilde{O} ) \otimes \mathbb I \right) U^\dagger \lvert b\rangle^2. 
\end{align}

The last equality follows from the fact that $\mathcal M_{\ensemble, \channel}^{-1, \dagger}$ factorizes and is unital. We write $U \in \ensemble = \ensemble_1 \otimes \ldots \otimes \ensemble_n$ as $V \otimes W$, with $V = U_1 \otimes \ldots \otimes U_k, W = U_{k+1} \otimes \ldots \otimes U_n$. We also write $\ensemble_1 \otimes \ldots \otimes \ensemble_k$ and $\ensemble_{k+1} \otimes \ldots \otimes \ensemble_n$ as $\ensemble_{1\ldots k}$ and $\ensemble_{k+1 \ldots n}$, respectively. The expression becomes 

\begin{align}
&\begin{aligned}
    &= \max_{\sigma \in \mathbb{D}_{2^n}} \underset{V \sim \ensemble_{1 \ldots k}}{\E} \underset{W \sim \ensemble_{k+1 \ldots n}}{\E} \sum_{c \in \{0,1\}^k} \sum_{d \in \{0,1\}^{(n-k)}} \langle c \rvert\! \langle d \rvert \channel(V \otimes W \sigma V^\dagger \otimes W^\dagger) \lvert c \rangle\!\lvert d \rangle \\ &\qquad \cdot \langle c \rvert \!\langle d \rvert (V \otimes W) \left( \bigotimes_{i = 1}^k \mathcal{M}_{\ensemble_i, \channel_i}^{-1, \dagger}(\tilde{O} ) \otimes \mathbb I \right) (V^\dagger \otimes W^\dagger)\lvert c\rangle\!\lvert d\rangle^2 
    \end{aligned} \nn
    &\begin{aligned}
    &= \max_{\sigma \in \mathbb{D}_{2^n}} \underset{V \sim \ensemble_{1 \ldots k}}{\E} \underset{W \sim \ensemble_{k+1 \ldots n}}{\E} \sum_{c \in \{0,1\}^k} \sum_{d \in \{0,1\}^{(n-k)}} \langle c \rvert\! \langle d \rvert \channel(V \otimes W \sigma V^\dagger \otimes W^\dagger) \lvert c \rangle\!\lvert d \rangle \\ &\qquad \cdot \langle c \rvert V  \bigotimes_{i = 1}^k \mathcal{M}_{\ensemble_i, \channel_i}^{-1, \dagger}(\tilde{O} ) V^\dagger \lvert c\rangle^2 
    \end{aligned} \nn
    &\begin{aligned}
    &= \max_{\sigma \in \mathbb{D}_{2^n}}     \underset{V \sim \ensemble_{1 \ldots k}}{\E} \sum_{c \in \{0,1\}^k} \langle c \rvert V  \bigotimes_{i = 1}^k \mathcal{M}_{\ensemble_i, \channel_i}^{-1, \dagger}(\tilde{O} ) V^\dagger \lvert c\rangle^2 \\ &\qquad \cdot \langle c \rvert \left( \underset{W \sim \ensemble_{k+1 \ldots n}}{\E} \sum_{d \in \{0,1\}^{(n-k)}} \! \langle d \rvert \channel(V \otimes W \sigma V^\dagger \otimes W^\dagger) \lvert d \rangle \right) \lvert c \rangle.
    \end{aligned}
\end{align}

To simplify the expression further, we focus on the summation over $d \in \{0,1\}^{n-k}$, which is exactly the partial trace over the last $n-k$ qubits. 
We denote the partial trace over the last $n-k$ qubits as $\tr_{k+1 \ldots n}$. 
We write $\sigma = \sum_a E_a \otimes F_a$, where $E_a \in \mathcal L(\mathbb C^{2^n}), F_a \in \mathcal L(\mathbb C^{2^{n-k}})$, and $\channel = \channel_1 \otimes \ldots \otimes \channel_n = \channel_{1 \ldots k} \otimes \channel_{k+1 \ldots n}$, where $\channel_{1 \ldots k} = \channel_1 \otimes \ldots \otimes \channel_k, \channel_{k+1 \ldots n} = \channel_{k+1} \otimes \ldots \otimes \channel_n$.

\begin{align}
 \sum_{d \in \{0,1\}^{(n-k)}} \! \langle d \rvert \channel(V \otimes W \sigma V^\dagger \otimes W^\dagger) \lvert d \rangle
 &= \tr_{k+1 \ldots n} \channel(V \otimes W \sigma V^\dagger \otimes W^\dagger) \nn 
 &= \sum_a \tr_{k+1 \ldots n}\big( \channel_{1\ldots k} (V E_a V^\dagger)\otimes \channel_{k+1 \ldots n}(W F_a W^\dagger) \big) \nn
 &= \sum_a \channel_{1\ldots k} (V E_a V^\dagger) \tr\big(\channel_{k+1 \ldots n}(W F_a W^\dagger)\big) \nn 
 &= \sum_a \channel_{1\ldots k} (V E_a V^\dagger) \tr(W F_a W^\dagger) \nn
 &= \sum_a \channel_{1\ldots k} (V E_a V^\dagger) \tr(F_a) \nn
 &= \channel_{1\ldots k} (V \sum_a E_a \tr(F_a) V^\dagger) \nn
 &= \channel_{1\ldots k} (V \tr_{k+1}(\sigma) V^\dagger).
\end{align}

Plugging into the expression for the shadow seminorm, we get 
\begin{align}
    &= \max_{\sigma \in \mathbb{D}_{2^n}}     \underset{V \sim \ensemble_{1 \ldots k}}{\E} \sum_{c \in \{0,1\}^k} \langle c \rvert V  \bigotimes_{i = 1}^k \mathcal{M}_{\ensemble_i, \channel_i}^{-1, \dagger}(\tilde{O} ) V^\dagger \lvert c\rangle^2  \langle c \rvert \Big( \underset{W \sim \ensemble_{k+1 \ldots n}}{\E} \channel_{1\ldots k} (V \tr_{k+1}(\sigma) V^\dagger)\Big) \lvert c \rangle \nn
    &= \max_{\sigma \in \mathbb{D}_{2^n}}     \underset{V \sim \ensemble_{1 \ldots k}}{\E} \sum_{c \in \{0,1\}^k} \langle c \rvert \channel_{1\ldots k} (V \tr_{k+1}(\sigma) V^\dagger) \lvert c \rangle\!\langle c \rvert V \mathcal{M}_{\ensemble_{1 \ldots k}, \channel_{1 \ldots k}}^{-1, \dagger}(\tilde{O} ) V^\dagger \lvert c\rangle^2.  
\end{align}
Because the partial trace preserves the space of quantum states, 
\begin{align}
  \norm{O}_{\mathrm{shadow},\ensemble, \channel}^2   
    &= \max_{\tau \in \mathbb{D}_{2^k}}     \underset{V \sim \ensemble_{1 \ldots k}}{\E} \sum_{c \in \{0,1\}^k} \langle c \rvert \channel_{1\ldots k} (V \tau V^\dagger) \lvert c \rangle\!\langle c \rvert V \mathcal{M}_{\ensemble_{1 \ldots k}, \channel_{1 \ldots k}}^{-1, \dagger}(\tilde{O} ) V^\dagger \lvert c\rangle^2  \nn
    &= \norm{\tilde{O}}_{\mathrm{shadow},\ensemble_{1 \ldots k}, \channel_{1 \ldots k}}^2.
\end{align}
\end{proof}

\section{Global Clifford Ensemble with Noise}\label{section:global_clifford}
In this section we prove that if $\mathcal U$ is the Clifford group and $\mathcal E$ is an arbitrary quantum channel, the shadow channel can be expressed as a depolarizing channel (\cref{def:depolarizing-channel}).
In this setting, we derive the expression for the classical shadow and the sample complexity of the classical shadows protocol. 

\subsection{Derivation of Shadow Channel}\label{subsec:global-derivation-of-shadow-channel}
We begin with a technical lemma involving $2$-design ensembles. 
\begin{lemma}\label{lemma:2design}
Let $\Udag$ be an $n$-qubit $2$-design and $\channel: \mathcal L (\mathbb{C}^{2^n}) \rightarrow \mathcal L (\mathbb{C}^{2^n})$ be a linear superoperator. Then, 

\begin{align}\label{eq:shadow_channel_2design}
    \mathcal{M}_{\ensemble, \channel} (A) = f(\channel) A + \Big(\frac{1}{2^n}\tr(\channel(\mathbb{I})) - f(\channel) \Big) \tr(A)\frac{\mathbb{I}}{2^n}, 
\end{align}
where 
\begin{align}\label{eq:fidelity_2design}
    f(\channel) = \frac{1}{2^{2n}-1} \Big( \tr(\channel \circ \mathrm{diag}) - \frac{1}{2^n} \tr(\channel (\mathbb I)) \Big).
\end{align}
Also, if $\channel$ is trace-preserving or unital, then, 
\begin{align}
    \mathcal M_{\mathcal U, \mathcal E} = \mathcal D_{n,f(\mathcal E)},\qquad \text{where} \quad
    f(\mathcal E) = \frac{\tr(\mathcal E\circ \mathrm{diag})-1}{2^{2n}-1}.
\end{align}
\end{lemma}
\begin{proof}
We first introduce the following notation:
let $\mathcal E: \mathcal L(\mathbb C^d) \rightarrow L(\mathbb C^d)$ be a linear superoperator. Define the unary operator $()^\ddagger$ as follows:
\begin{align}
\label{def:ddagger}
    \mathcal E^\ddagger(A) &= (\mathcal E^*(A^\dag))^\dag. 
\end{align}
Say that $\mathcal E$ has Kraus representation
\begin{align}
    \mathcal E:B\mapsto \sum_i J_i B K_i^\dag.
\end{align} 
Then, the shadow channel may be evaluated as
\begin{align}
    \shadowchannel{\ensemble}{\channel}(A) &=
    \E\limits_{U\sim \ensemble} \sum_{b\in \{0,1\}^n}
    \bra b \channel(UAU^\dag)\ket b U^\dag \ketbra{b}{b} U \nonumber\\
    &= \E\limits_{U\sim \ensemble} \sum_{b\in \{0,1\}^n}
    \bra b \left(\sum_i J_i UAU^\dag K_i^\dag \right) \ket b U^\dag \ketbra{b}{b} U \nonumber\\
    &=
    \sum_{b \in \{0,1\}^n} \E\limits_{U\sim \ensemble}
    \tr\left(U^\dag  \sum_i K_i^\dag \ketbra bb J_i UA \right) U^\dag \ketbra bb U 
    \nn
    &= 
    \sumCube{b} \E\limits_{U\sim \ensemble} \tr_1
    \left(
    U^\dag \mathcal E^\ddagger(\ketbra bb) UA \otimes U^\dag \ketbra bb U
    \right)
    \nn
    &=
    \sumCube{b} \tr_1 \left\{
    \E\limits_{U\sim \ensemble} (U^\dag \otimes U^\dag)(\mathcal E^\ddagger(\ketbra bb)\otimes \ketbra bb)(U\otimes U)(A\otimes I)
    \right\}
    \nn
    &=
    \sumCube{b} \tr_1 \left\{
    \E\limits_{U\sim \ensemble^\dag} (U \otimes U)(\mathcal E^\ddagger(\ketbra bb)\otimes \ketbra bb)(U^\dag\otimes U^\dag)(A\otimes I)
    \right\}
    \nn
    &=
    \sumCube{b}
    \tr_1 \Bigg\{
    \underbrace{
    T_2^{(2^n)}\left(\mathcal E^\ddagger(\ketbra bb) \otimes \ketbra bb \right)
    }_{\circled{1}}
    (A\otimes I)
    \Bigg\},
\end{align}
since $\Udag$ is a 2-design. In the above equations, $\tr_1$ denotes the partial trace over the first subsystem.

Applying \cref{eq:HaarIntegralt2} to the 2-fold twirl $\circled{1}$ gives

\begin{align}
\circled{1} &=  \int \dd\eta(U) \ (U\otimes U)
    \left[ \mathcal E^\ddagger(\ketbra bb) \otimes \ketbra bb \right](U^\dag \otimes U^\dag) \nn
    &=
    \frac{1}{2^{2n}-1} \left[
    \tr\left(
    \mathcal E^\ddagger(\ketbra bb) 
    \otimes \ketbra bb\right)\left(I-\frac{W}{2^n}\right) +
    \tr\left(
    W \mathcal E^\ddagger(\ketbra bb) 
    \otimes \ketbra bb \right)\left(W-\frac{I}{2^n}\right)
    \right],
\end{align}
where the traces in the above equations simplify as 
\begin{align}
\tr\left(
    \mathcal E^\ddagger(\ketbra bb) 
    \otimes \ketbra bb\right) 
    = \tr(\mathcal E^\ddagger(\ketbra bb)), \\
    \tr\left(
    W\mathcal E^\ddagger(\ketbra bb) 
    \otimes \ketbra bb\right) = 
    \bra b \mathcal E^\ddagger(\ketbra bb) \ket b.
\end{align}
Therefore,
\begin{align}
\label{eq:MUEAexpansion}
    \shadowchannel{\ensemble}{\channel}(A) &=
    \sumCube{b} 
    \tr_1 \left\{
    \frac{1}{2^{2n}-1} \left(
    \tr( 
    \mathcal E^\ddagger (\ketbra bb)) \left(I-\frac{W}{2^n}\right)
    +
    \bra b \mathcal E^\ddagger(\ketbra bb) \ket b \left(W-\frac{I}{2^n}\right)
    \right)(A\otimes I)
    \right\}
    \nn
    &=
    \frac{1}{2^{2n}-1}
    \Bigg[
    \underbrace{\tr(\mathcal E^\ddagger(I))}_{\circled{2}} \underbrace{\tr_1\left\{
    \left(
    I - \frac{W}{2^n}
    \right)(A\otimes I)\right\}}_{\circled{3}}
    \nn
    &\qquad\qquad\qquad +
    \underbrace{\sumCube{b} \bra b \mathcal E^\ddagger(\ketbra bb) \ket b
    }_{\circled{4}}
    \underbrace{\tr_1\left\{
    \left(
    W - \frac{I}{2^n}
    \right)(A\otimes I)\right\}}_{\circled{5}}
    \Bigg].
\end{align}

Then, by simple calculation, 
\begin{align}
    \circled{2} = \tr(\mathcal E^\ddagger(I)) = \tr(\mathcal E(I)),
\end{align}
and
\begin{align}
    \circled{4} = \Tr(\mathcal E^\ddagger \circ \diag)
    = \Tr(\mathcal E\circ \diag).
\end{align}
To evaluate $\circled{3}$ and $\circled{5}$, we use the fact that
\begin{align}
    \tr_1 (W(A\otimes I)) = A.
\end{align}
Hence,
\begin{align}
    \circled{3} &= \tr_1(A\otimes I) - \frac{1}{2^n} \tr_1(W(A\otimes I))
    \nn
    &=
    \tr(A) I - \frac{1}{2^n} A
\end{align}
and
\begin{align}
    \circled{5} &= \tr_1(W(A\otimes I)) - 
    \frac{1}{2^n} \tr_1(A\otimes I) \nn
    &= A - \frac{1}{2^n} \tr(A) I.
\end{align}
Plugging these back into \cref{eq:MUEAexpansion}, 
\begin{align}
    \shadowchannel{\ensemble}{\channel}(A)
    &=
    \frac{1}{2^{2n}-1} \left[
    \tr(\mathcal E(I))(\tr(A) I -\frac{1}{2^n} A)
    +
    \tr(\mathcal E\circ \diag)
    (A- \frac{1}{2^n} \tr(A) I)
    \right]
    \nn
    &= f(\channel) A + \Big(\frac{1}{2^n}\tr(\channel(\mathbb{I})) - f(\channel) \Big) \tr(A)\frac{\mathbb{I}}{2^n},
\end{align}
where
\begin{align}
    f(\channel) = \frac{1}{2^{2n}-1} \Big( \tr(\channel \circ \mathrm{diag}) - \frac{1}{2^n} \tr(\channel (\mathbb I)) \Big),
\end{align}
which completes the first part of the proof. 
The second part of the proof follows from the fact that, if $\channel$ is trace-preserving or unital, then $\tr(\channel(\mathbb I)) = \tr(\mathbb I) = 2^n$. Substituting this into \cref{eq:shadow_channel_2design} and \cref{eq:fidelity_2design} gives the result. 
\end{proof}

Since the Clifford group forms a 2-design, \cref{lemma:2design} immediately implies that if $\mathcal E$ is an arbitrary quantum channel (and is hence trace-preserving), then the shadow channel is given by
\begin{align}
    \mathcal M_{\mathcal C_n, \mathcal E} = \mathcal D_{n,f(\mathcal E)},
    \qquad
    \text{where}
    \quad
    f(\mathcal E) = \frac{\tr(\mathcal E\circ \mathrm{diag})-1}{2^{2n}-1}.
    \label{eq:clifford_global_noise}
\end{align}
Moreover, if the shadow channel is invertible, its inverse is given by $\mathcal M_{\mathcal C_n, \mathcal E}^{-1} = \mathcal D_{n,f(\mathcal E)^{-1}}$.

An important fact can be deduced from \cref{lemma:2design}, namely, that if the unitary ensemble is a $2$-design, then the shadow channel is invertible if and only if the error channel obeys the simple condition $\tr(\mathcal E \circ \diag) \neq 1$: 
\begin{claim}\label{claim:invertibility_global_clifford}
Let $\Udag$ be an $n$-qubit 2-design, and let $\channel$ be a linear superoperator. $\mathcal M_{\ensemble, \channel}$ is invertible if and only if $\tr(\channel \circ \mathrm{diag}) \neq 1$.
In this case, 
$\mathcal M_{\ensemble, \channel}^{-1} = \mathcal D_{n, 1/f(\channel)}$.
\end{claim}
\begin{proof}
By \cref{lemma:2design}, the shadow channel with noise is a depolarizing channel. Therefore, 
\begin{align}
\text{$\mathcal{M}_{\ensemble, \channel} = \mathcal{D}_{n, f(\channel)}$ is invertible} \iff f(\channel) \neq 0 \iff \tr(\channel \circ \mathrm{diag}) \neq 1, 
\end{align}
and
\begin{align}
\mathcal M_{\ensemble, \channel}^{-1} = \mathcal D_{n, f(\channel)}^{-1} = \mathcal D_{n, 1/f(\channel)}.
\end{align}
\end{proof}

Next, we prove bounds on the depolarizing parameter $f(\channel)$:

\begin{claim}\label{claim:fidelity-parameter-bounds}
Let $\channel: \mathcal{L}({\mathbb{C}^{2^n}}) \rightarrow \mathcal{L}({\mathbb{C}^{2^n}})$ be a quantum channel. Then, 

\begin{align}
    -\frac{1}{2^{2n}-1} \leq f(\channel) \leq \frac{1}{2^n + 1}.
    \label{eq:bounds_on_fE1}
\end{align}
\label{claim:bounds_on_fE}
\end{claim}

\begin{proof}
\begin{align}
    \text{$\channel$ is a quantum channel} &\implies \tr(\channel \circ \mathrm{diag}) \in [0, 2^n]\nn
    &\implies f(\channel) = \frac{\tr(\channel \circ \mathrm{diag}) - 1}{2^{2n}-1} \in \Big[-\frac{1}{2^{2n}-1}, \frac{1}{2^n + 1}\Big].
    \label{eq:bounds_on_fE}
\end{align}
\end{proof}

A few remarks are in order. First, note that the bounds in \cref{claim:fidelity-parameter-bounds} have appeared in work on randomized benchmarking (e.g., Lemma 1 in \cite{Wallman2018randomized}).

Second, note that the depolarizing parameter $f(\channel)$ is upper bounded by $\frac{1}{2^n + 1} = f(\mathbb{I})$, which is the depolarizing parameter of the noiseless shadow channel $\mathcal{M}_\ensemble$. In other words, as expected, noise necessarily decreases the depolarizing parameter of the shadow channel (i.e., it is not possible to use noise to improve the performance of classical shadows).

Finally, note from \cref{eq:bounds_on_fE1} that $f(\mathcal E)$ can take negative values. As discussed in \cref{sec:linear_operators_and_quantum_channels}, while it is typical to consider depolarizing channels with depolarizing parameter $f\in[0,1]$, $D_{n,f}$ remains a quantum channel for some negative values of $f$. We note here that the lower bound in \cref{eq:bounds_on_fE1} matches exactly the lower bound in \cref{eq:when_is_depolarizing_channel_a_channel} when $d=2^n$.

\subsection{Classical Shadow}
We now give an expression for the classical shadow when $\mathcal M_{\ensemble, \channel} = \mathcal{D}_{n, f(\channel)}$. Recall that the classical shadow corresponding to a unitary ensemble $\ensemble$, noise channel $\channel$, unitary transformation $\hat{U} \in \ensemble$, and bit string $\hat{b} \in \{0,1\}^n$ is $\hat{\rho}(\ensemble, \channel, \hat{U}, \hat{b}) = \mathcal M_{\ensemble, \channel}^{-1}(\hat{U}^\dagger \lvert \hat{b}\rangle\!\langle \hat{b}\rvert \hat{U})$.

\begin{claim}\label{claim:global_clifford_classical_shadow}
Let $\Udag$ be an $n$-qubit 2-design, and let $\channel$ be a quantum channel. Assume $\mathcal M_{\ensemble, \channel}$ is invertible. Then, for some $\hat{U} \in \ensemble$ and $\hat{b} \in \{0, 1\}^n$, the classical shadow is

\begin{align}
    \hat{\rho}(\ensemble, \channel, \hat{U}, \hat{b}) = \frac{1}{f(\channel)} \hat{U}\lvert\hat{b}\rangle\!\langle\hat{b}\rvert + \Big(1 - \frac{1}{f(\channel)}\Big)\frac{\mathbb{I}}{2^n},
    \end{align}
    where $f(\mathcal E) = \frac{\tr(\mathcal E\circ \mathrm{diag})-1}{2^{2n}-1}$.
\end{claim}
\begin{proof}
\begin{align}
    \hat{\rho}(\ensemble, \channel, \hat{U}, \hat{b}) 
    &=  \mathcal M_{\ensemble, \channel}^{-1}(\hat{U}^\dagger \lvert \hat{b}\rangle\!\langle \hat{b}\rvert \hat{U}) \nn
    &=  \mathcal D_{n, f(\channel)}^{-1}(\hat{U}^\dagger \lvert \hat{b}\rangle\!\langle \hat{b}\rvert \hat{U}) \nn
    &=  \mathcal D_{n, 1/f(\channel)}(\hat{U}^\dagger \lvert \hat{b}\rangle\!\langle \hat{b}\rvert \hat{U}) \nn
    &= \frac{1}{f(\channel)} \hat{U}\lvert\hat{b}\rangle\!\langle\hat{b}\rvert \hat{U} + \Big(1 - \frac{1}{f(\channel)}\Big)\frac{\mathbb{I}}{2^n}.
\end{align}
\end{proof}

\subsection{Derivation of Shadow Seminorm}\label{subsec:global-derivation-of-shadow-norm}
We derive an expression for the shadow seminorm of a traceless observable. Recall that the sample complexity to estimate $\tr(O\rho)$ for some observable $O$ is upper bounded by the shadow seminorm of the traceless part of $O$ (\cref{lemma:upper-bound-on-variance}), which we write as $O_o$. 

\begin{proposition}\label{prop:global_shadow_norm}
Let $\Udag$ be an $n$-qubit 3-design, and let $\channel$ be a linear superoperator, and let $O_o$ be a traceless observable. Then, 
\begin{align}
\norm{O_o}_{\shadow, \ensemble, \channel}^2 = \frac{d(d^2 -1)}{(d+2)(d\beta - \alpha)} \bigg( \frac{(1 + d)\alpha - 2\beta}{d\beta - \alpha} \tr(O_o^2) + 2 \norm{O_o^2}_{\mathrm{sp}}  \bigg)
\end{align}
where 
\begin{align}
d = 2^n,\quad \alpha = \tr(\channel(\mathbb{I})),\quad \beta = \tr(\channel \circ \mathrm{diag}).
\end{align}
Also, if  $\channel$ be a trace-preserving or unital linear superoperator, then the expression simplifies to
\begin{align}
\norm{O_o}_{\shadow, \ensemble, \channel}^2 = \frac{d^2 -1}{(d+2)(\beta - 1)} \bigg( \frac{d + d^2 - 2\beta}{d(\beta - 1)} \tr(O_o^2) + 2 \norm{O_o^2}_{\mathrm{sp}}  \bigg)
\end{align}
where 
\begin{align}
d = 2^n,\quad \quad \beta = \tr(\channel \circ \mathrm{diag}).
\end{align}
\end{proposition}
\begin{proof}
\begin{align}
\norm{O_o}_{\shadow, \ensemble, \channel}^2 
&= \max_{\sigma \in \mathbb{D}_{2^n}} \underset{U \sim \ensemble}{\E} \sum_{b \in \{0,1\}^n} \langle b \rvert \channel(U \sigma U^\dagger) \lvert b \rangle\!\langle b \rvert U \mathcal{M}^{-1,\dagger}_{\ensemble, \channel}(O_o) U^\dagger \lvert b \rangle^2 \nn 
&= \max_{\sigma \in \mathbb{D}_{2^n}} \underset{U \sim \ensemble}{\E} \sum_{b \in \{0,1\}^n} \langle b \rvert \channel(U \sigma U^\dagger) \lvert b \rangle\!\langle b \rvert U \mathcal{D}_{n, 1/f(\channel)}(O_o) U^\dagger \lvert b \rangle^2 \nn
&= \frac{1}{f(\channel)^2} \max_{\sigma \in \mathbb{D}_{2^n}} \underset{U \sim \ensemble}{\E} \sum_{b \in \{0,1\}^n} \langle b \rvert \channel(U \sigma U^\dagger) \lvert b \rangle\!\langle b \rvert U O_o U^\dagger \lvert b \rangle^2 \nn
&= \frac{1}{f(\channel)^2} \bigg(\frac{(1+d)\alpha - 2\beta}{(d-1)d(d+1)(d+2)}\tr(O_o^2) + \frac{2(d\beta -\alpha)}{(d-1)d(d+1)(d+2)} \norm{O_o^2}_{\mathrm{sp}} \bigg). 
\end{align}
The final equality follows from \cref{lem:identity3designsum}. To simplify further, we get an expression for $1/f(\channel)^2$. 
\begin{align}
    f(\channel) 
    &= \frac{1}{2^{2n}-1}\big(\tr(\channel \circ \mathrm{diag}) - \frac{1}{2^n}\tr(\channel(\mathbb{I})) \big) \nn 
    &= \frac{d\beta - \alpha}{(d-1)d(d+1)} \nn
    \implies \frac{1}{f(\channel)^2} &= \frac{(d-1)^2d^2(d+1)^2}{(d\beta - \alpha)^2}.
\end{align}

Plugging into the expression above, we get

\begin{align}
& \norm{O_o}_{\shadow, \ensemble, \channel}^2 \nonumber\\ 
&= \frac{(d-1)^2d^2(d+1)^2}{(d\beta - \alpha)^2}\bigg(\frac{(1+d)\alpha - 2\beta}{(d-1)d(d+1)(d+2)}\tr(O_o^2) + \frac{2(d\beta -\alpha)}{(d-1)d(d+1)(d+2)} \norm{O_o^2}_{\mathrm{sp}} \bigg)\nn 
&= \frac{d(d^2-1)}{d + 2} \cdot \frac{(1+d)\alpha - 2\beta}{(d\beta - \alpha)^2} \tr(O_o^2) + \frac{d(d^2-1)}{d + 2} \cdot \frac{2}{d\beta - \alpha}  \norm{O_o^2}_{\mathrm{sp}}\nn
&= \frac{d(d^2-1)}{(d + 2)(d\beta - \alpha)} \bigg(\frac{(1+d)\alpha - 2\beta}{d\beta - \alpha} \tr(O_o^2) + 2 \norm{O_o^2}_{\mathrm{sp}}\bigg).
\end{align}

Finally, if $\channel$ is trace-preserving or unital, then $\alpha = \tr(\channel(\mathbb I)) = d$. The second part of the proposition follows from substituting $\alpha = d$. 
\end{proof}

Building from \cref{prop:global_shadow_norm}, one can get looser bounds that are more convenient to work with. 
 
\begin{corollary}\label{prop:global_shadow_norm_bounds}
Let $\ensemble$ be an $n$-qubit 3-design, let $\channel$ be a trace-preserving or unital linear superoperator, and let $O_o$ be a traceless observable. Then, 
\begin{align}
\frac{(2^{n} - 1)^2}{(\beta - 1)^2}  \tr(O_o^2) \leq \norm{O_o}_{\shadow, \ensemble, \channel} \leq \frac{3(2^{n}-1)^2}{(\beta - 1)^2}\tr(O_o^2)
\leq \frac{3(2^n-1)^2}{(\beta - 1)^2}\tr(O^2).
\end{align}
where $\beta = \tr(\channel \circ \mathrm{diag})$.
\end{corollary}

The proof is straightforward. We include it in \cref{app:proof-of-global-shadow-norm-bounds} for completeness.

Combining \cref{thm:general-theorem} and \cref{prop:global_shadow_norm_bounds} yields sample complexity bounds on the classical shadows protocol when the unitary ensemble is the Clifford group (or, any unitary $3$-design).

\begin{corollary}
\label{cor:sampleComplexityGlobal}
Let $\{O_i\}_{i=1}^M$ be a collection of $M$ observables. Let $\ensemble$ be a unitary 3-design (e.g., the Clifford group).
Let $\channel: \mathcal{L}(\mathbb C^{2^n}) \rightarrow \mathcal{L}(\mathbb C^{2^n})$ be a quantum channel such that $\tr(\channel \circ \mathrm{diag}) \neq 1$. 
The sample complexity $N_\mathrm{tot}$ to estimate the linear target functions $\{\tr(O_i \rho)\}_{i=1}^M$ of an $n$-qubit state $\rho$ within error $\eps$ and failure probability $\delta$ when the unitary ensemble $\ensemble$ is subject to the error channel $\channel$ is  

\[
N_\mathrm{tot} \leq \frac{204(2^n -1)^2\log(2M/\delta)}{(\beta - 1)^2\eps^2} \max_{1 \leq i \leq M} \tr(O_i^2),
\]
where $\beta = \tr(\channel \circ \mathrm{diag})$.
\end{corollary}

\subsection{Examples}\label{subsec:global-examples}
\cref{eq:clifford_global_noise} establishes that if $\mathcal U$ is the Clifford group and $\channel$ is an arbitrary quantum channel, then the resulting shadow channel $\mathcal{M}_{\mathcal{C}_n, \channel}$ is always a depolarizing channel (\cref{def:depolarizing-channel}) with a depolarizing parameter that depends on the quantum channel $\channel$. 
\cref{cor:sampleComplexityGlobal} establishes the sample complexity in this scenario. 
We now apply our results to derive expressions for the shadow channel, inverse shadow channel, classical shadow, shadow seminorm, and sample complexity for the classical shadows protocol with the Clifford group and specific quantum channels.

\subsubsection{Noiseless Case}\label{subsec:global_noiseless_case}
We begin with a basic example, the case where the quantum channel is the identity channel, to show that the results from \cite{huang2020predicting} can be recovered. For reference, see Eqs.~(S37) through (S43) in \cite{huang2020predicting}). 

\begin{claim}
\label{claim:noiselessShadows}
In the noiseless case, 
\begin{enumerate}
    \item \label{item:noiseless1} $f(\mathbb I) = \frac{1}{2^n + 1}.$
\item \label{item:noiseless2} $\mathcal M_{\mathcal C_n, \mathbb I} = \mathcal{D}_{n, 1/2^n + 1}$.
\item \label{item:noiseless3} $\mathcal M_{\mathcal C_n, \mathbb I}^{-1} = \mathcal{D}_{n, 2^n + 1}.$
\item \label{item:noiseless4} The classical shadow can be written as $\hat{\rho}(\mathcal C_n, \mathbb I, \hat{U}, \hat b) = (2^n + 1) \hat{U}^\dagger \lvert \hat b\rangle\!\langle \hat b \rvert \hat U - \mathbb I$.
\item \label{item:noiseless5}
$\norm{O_o}_{\shadow, \mathcal C_n, \mathbb I}^2 = \frac{2^n + 1}{2^n + 2}\big(\tr(O_o^2) + 2 \norm{O_o^2}_\mathrm{sp}\big).$
\end{enumerate}

\end{claim}

\begin{proof}
(\ref{item:noiseless1}) follows from a simple calculation:
\begin{align}
    f(\mathbb I) = \frac{\tr(\mathbb I \circ \mathrm{diag})-1}{2^{2n}-1} = \frac{2^n-1}{2^{2n}-1} 
    = \frac{1}{2^n + 1}.
\end{align}
(\ref{item:noiseless2}) and (\ref{item:noiseless3}) follow from \cref{eq:clifford_global_noise}.
To prove (\ref{item:noiseless4}), apply \cref{claim:global_clifford_classical_shadow}:
\begin{align}
\hat{\rho}(\mathcal C_n, \mathbb I, \hat{U}, \hat b) 
&=  (2^n + 1) \hat{U}^\dagger \lvert \hat b\rangle\!\langle \hat b \rvert \hat U + (1 - 2^n + 1) \frac{\mathbb I}{2^n} \nn
&= (2^n + 1) \hat{U}^\dagger \lvert \hat b\rangle\!\langle \hat b \rvert \hat U - \mathbb I.
\end{align}
To prove (\ref{item:noiseless5}), apply \cref{prop:global_shadow_norm} with $\beta = \tr(\mathbb I \circ \mathrm{diag}) = 2^n = d$. Then,  
\begin{align}
    \norm{O_o}^2_{\shadow, \mathcal C_n, \mathbb I} 
    &= \frac{d^2 - 1}{(d+2)(d-1)}\Bigg( \frac{d + d^2 -2d}{d(d-1)} \tr(O_o^2) + 2 \norm{O_o^2}_\mathrm{sp} \bigg) \nn
    &= \frac{2^n + 1}{2^n + 2}\bigg(\tr(O_o^2) + 2 \norm{O_o^2}_\mathrm{sp}\bigg).
\end{align}
\end{proof}

\begin{remark}
One can verify that the dephasing channel is an inconsequential noise channel (see \cref{claim:inconsequential_noise} of \cref{app:inconsequential-noise}). As such, these results also hold when $\channel$ is the dephasing channel.
\end{remark}

\subsubsection{Depolarizing Channel}\label{subsec:global_depolarizing}

We derive expressions for the shadow channel, inverse shadow channel, the classical shadow, and the shadow seminorm when the Clifford group $\mathcal C_n$ is subject to depolarizing noise with depolarizing parameter $f$ (\cref{def:depolarizing-channel}).

\begin{claim}
\label{claim:depolar_channel}
If the unitary ensemble used in the classical shadows protocol is the Clifford group and is subject to depolarizing noise with depolarizing parameter $f \in [0,1]$, then 
\begin{enumerate}
    \item \label{item:dep1} $f( \mathcal{D}_{n, f}) = \frac{f}{2^n + 1}$.
    \item \label{item:dep2} $\mathcal M_{\mathcal C_n, \mathcal{D}_{n, f}} = \mathcal{D}_{n, f/2^n + 1}$.
    \item \label{item:dep3}  $\mathcal M_{\mathcal C_n, \mathcal D_{n,f}}^{-1} = \mathcal{D}_{n, 2^n + 1/ f}.$ 
    \item \label{item:dep4} The classical shadow can be written as 
$\hat{\rho}(\mathcal C_n, \mathcal D_{n, f}, \hat{U}, \hat b) = \frac{2^n + 1}{f} \hat{U}^\dagger \lvert \hat b\rangle\!\langle \hat b \rvert \hat U - (1 - \frac{2^n +1}{f}) \frac{\mathbb I}{2^n}$.
\item \label{item:dep5}
Let $O \in \mathbb H_{2^n}$. Then,$ 
\norm{O - \frac{1}{2^n}\tr(O) \mathbb I}^2_{\shadow, \mathcal C_n, \mathcal D_{n,f}} \leq \frac{3}{f^2}\tr(O^2).$
\end{enumerate}

\end{claim}

\begin{proof}
First we prove (\ref{item:dep1}):
\begin{align}
    \tr(\mathcal{D}_{n,f} \circ \mathrm{diag}) 
    &= \sum_{b\in\{0,1\}^n} \langle b \rvert \mathcal{D}_{n,f}(\lvert b\rangle\!\langle b\rvert) \lvert b \rangle \nn
    &= \sum_{b\in\{0,1\}^n} \langle b \rvert \big(f\lvert b\rangle\!\langle b\rvert + (1-f) \frac{\mathbb I}{2^n} \big) \lvert b \rangle \nn
    &= \sum_{b\in\{0,1\}^n}  f + (1-f) \frac{1}{2^n} \nn
    &= 2^nf + 1-f.
\end{align}
Then, 
\begin{align}
    f(\channel)  
    &= \frac{\tr(\mathcal{D}_{n,f} \circ \mathrm{diag})-1}{2^{2n}-1} \nn
    &= \frac{2^nf +1-f -1}{2^{2n}-1} \nn
    &= \frac{f(2^n -1)}{2^{2n}-1} \nn
    &= \frac{f}{2^{n}+1}.
\end{align}

(\ref{item:dep2}) and (\ref{item:dep3}) follow from \cref{eq:clifford_global_noise}, and 
(\ref{item:dep5}) follows from \cref{prop:global_shadow_norm_bounds}.
To prove (\ref{item:dep4}), apply \cref{claim:global_clifford_classical_shadow}:
\begin{align}
\hat{\rho}(\mathcal C_n, \mathbb I, \hat{U}, \hat b) 
&= \frac{2^n + 1}{f} \hat{U}^\dagger \lvert \hat b\rangle\!\langle \hat b \rvert \hat U - \Big(1 - \frac{2^n +1}{f}\Big) \frac{\mathbb I}{2^n}.
\end{align}
\end{proof}

With a bound on the shadow seminorm in this setting, the sample complexity of the protocol follows.

\begin{corollary}
The sample complexity $N_\mathrm{tot}$ to estimate a collection of $M$ linear target functions $\tr(O_i \rho)$ within error $\eps$ and failure probability $\delta$ when the unitary ensemble $\ensemble$ is subject to depolarizing noise $ \mathcal{D}_{n, f} : A \mapsto fA + (1-f)\frac{\mathbb I}{2^n}$ is  

\[
N_\mathrm{tot} \leq   \frac{204\log(2M/\delta)}{f^2\eps^2} \max_{1 \leq i \leq M} \tr(O_i^2).
\]
\end{corollary}

\subsubsection{Amplitude Damping Channel}\label{subsec:global_amplitude_damping}

We derive expressions for the shadow channel, inverse shadow channel, the classical shadow, and the shadow seminorm when the Clifford group $\mathcal C_n$ is subject to the amplitude damping channel  (\cref{def:amplitude-damping-channel}).

\begin{claim}
\label{claim:amplitude_damping}
If the unitary ensemble used in the classical shadows protocol is the Clifford group and is subject to amplitude damping noise with parameter $p \in [0,1]$, then 
\begin{enumerate}
    \item \label{item:ad1}  $f(\mathrm{AD}_{n, p}) = \frac{(1+p)^n -1}{2^{2n} + 1}.$
    \item \label{item:ad2}  $\mathcal M_{\mathcal C_n, \mathrm{AD}_{n,p}} = \mathcal{D}_{n, ((1+p)^n - 1)/(2^{2n} - 1)}$.
    \item \label{item:ad3} $\mathcal M_{\mathcal C_n, \mathrm{AD}_{n,p}}^{-1} = \mathcal{D}_{n, (2^{2n} - 1)/ ((1+p)^n - 1)}$.
    \item \label{item:ad4} The classical shadow can be written as 
$\hat{\rho}(\mathcal C_n, \mathrm{AD}_{n,p}, \hat{U}, \hat b) = \frac{2^n + 1}{f} \hat{U}^\dagger \lvert \hat b\rangle\!\langle \hat b \rvert \hat U - (1 - \frac{2^n +1}{f}) \frac{\mathbb I}{2^n}$.
\item \label{item:ad5} 
Let $O \in \mathbb H_{2^n}$. Then, $\norm{O - \frac{1}{2^n}\tr(O) \mathbb I}^2_{\shadow, \mathcal C_n, \mathrm{AD}_{n,p}} \leq
\frac{3(2^{n}-1)^2}{((1+p)^n - 1)^2}\tr(O^2).$
\end{enumerate}
\end{claim}
\begin{proof}
To prove (\ref{item:ad1}), we use the fact that $ \langle 0 \rvert\mathrm{AD}_{1,p}(\lvert 0 \rangle\!\langle 0 \rvert) \lvert 0 \rangle = 1 $  
and $ \langle 1 \rvert\mathrm{AD}_{1,p}(\lvert 1 \rangle\!\langle 1 \rvert) \lvert 1 \rangle = p $. We denote the Hamming weight of a bit string $b$ as $\ell_1(b)$.

\begin{align}
    \tr(\mathrm{AD}_{n,p} \circ \mathrm{diag}) 
    &= \sum_{b\in\{0,1\}^n} \langle b \rvert \mathrm{AD}_{n,p}(\lvert b \rangle \! \langle b \rvert ) \lvert b \rangle \nn
    &= \sum_{b\in\{0,1\}^n} \prod_{i=1}^n \langle b_i \rvert \mathrm{AD}_{1,p}(\lvert b_i \rangle \! \langle b_i \rvert ) \lvert b_i \rangle \nn
    &= \sum_{b\in\{0,1\}^n} p^{\ell_1(b)} \nn
    &= \sum_{i=0}^n \binom{n}{i} p^i \nn
    &= (1+p)^n.
\end{align}

Then, 

\begin{align}
    f(\mathrm{AD}_{n,p})  
    &= \frac{\tr(\mathrm{AD}_{n,p} \circ \mathrm{diag})-1}{2^{2n}-1} \nn
    &= \frac{(1+p)^n-1}{2^{2n}-1}.
\end{align}

(\ref{item:ad2}) and (\ref{item:ad3}) follow from \cref{eq:clifford_global_noise}, and (\ref{item:ad5}) follows from \cref{prop:global_shadow_norm_bounds}.
To prove (\ref{item:ad4}), apply \cref{claim:global_clifford_classical_shadow},
\begin{align}
\hat{\rho}(\mathcal C_n, \mathrm{AD}_{n,p}, \hat{U}, \hat b) 
&= \mathcal D_{n, (2^{2n} - 1)/ ((1+p)^n - 1)}(\hat{U}^\dagger \lvert \hat b\rangle\!\langle \hat b \rvert \hat U)\nn 
&= \frac{2^{2n} - 1}{(1+p)^n - 1} \hat{U}^\dagger \lvert \hat b\rangle\!\langle \hat b \rvert \hat U - \Big(1 - \frac{2^{2n} - 1}{(1+p)^n - 1}\Big) \frac{\mathbb I}{2^n}.
\end{align}
\end{proof}

A bound on the sample complexity follows from the bound on the shadow seminorm.

\begin{corollary}
The sample complexity $N_\mathrm{tot}$ to estimate a collection of $M$ linear target functions $\tr(O_i \rho)$ within error $\eps$ and failure probability $\delta$ when the unitary ensemble $\ensemble$ is subject to the amplitude damping channel $ \mathrm{AD}_{n, p}$ is  

\[
N_\mathrm{tot} \leq   \frac{204(2^n -1)^2\log(2M/\delta)}{((1+p)^n - 1)^2\eps^2} \max_{1 \leq i \leq M} \tr(O_i^2).
\]
\end{corollary}

\section{Product Clifford Ensemble with Product Noise}
\label{sec:productClifford}
In this section we analyze the setting in which 
the quantum channel is a product channel (\cref{def:product-channel}) 
and 
the unitary ensemble is a product ensemble (\cref{def:product-ensemble}). 
Our results hold for any product ensemble in which each ensemble is a 3-design, which the product Clifford ensemble is an example.
We write the product Clifford ensemble as $\mathcal C_1^{\otimes n} = \{U_1 \otimes \ldots \otimes U_n : U_1, \ldots , U_n \in \mathcal C_1\}$ and the quantum channel as $\channel^{\otimes n}$, where $\channel$ acts on a single qubit.

\subsection{Derivation of Shadow Channel}

We derive expressions for the shadow channel and its inverse by building on the work in \cref{subsection:product_ensemble_noise} and \cref{section:global_clifford}. 

\begin{claim}\label{claim:shadow_channel_product_clifford}
Let $\ensemble = \ensemble_1 \otimes \ldots \otimes \ensemble_n$ be a product ensemble such that $\ensemble_i^\dagger$ is a 2-design for all $i \in [n]$. 
Let $\channel: \mathcal{L}(\mathbb C^{2^n}) \rightarrow \mathcal{L}(\mathbb C^{2^n})$ be a single-qubit quantum channel. Then, 
\begin{align}
    \mathcal M_{\ensemble , \channel^{\otimes n}} = \mathcal{D}_{1, \frac{1}{3}(\tr(\channel \circ \mathrm{diag}) - 1)}^{\otimes n}.
\end{align}
Also, if $\tr(\channel \circ \mathrm{diag}) \neq 1$, then 
\begin{align}
    \mathcal M_{\ensemble, \channel^{\otimes n}}^{-1} = \mathcal{D}_{1, 3/(\tr(\channel \circ \mathrm{diag}) - 1)}^{\otimes n}.
\end{align}
\end{claim}
\begin{proof}
The first part follows from \cref{claim:product_shadow_channel_factorizes} and \cref{lemma:2design}.
The second part follows from the fact that if $\tr(\channel \circ \mathrm{diag}) \neq 1$, then the shadow channel $\mathcal M_{\ensemble , \channel^{\otimes n}}$ is invertible (see \cref{claim:invertibility_global_clifford}). 
Therefore,
$\mathcal M_{\ensemble, \channel^{\otimes n}}^{-1} = ( \mathcal{D}_{1, (1/3)\cdot(\tr(\channel \circ \mathrm{diag}) - 1)}^{-1})^{\otimes n} = \mathcal{D}_{1, 3/(\tr(\channel \circ \mathrm{diag}) - 1)}$.
\end{proof}

\subsection{Classical Shadow}

We derive an expression for the classical shadow for the product Clifford ensemble and a product channel. 

\begin{claim}\label{claim:classical_shadow_product_clifford}
Let $\ensemble = \ensemble_1 \otimes \ldots \otimes \ensemble_n$ be a product ensemble such that $\ensemble_i^\dagger$ is a 2-design for all $i \in [n]$. 
Let $\channel: \mathcal{L}(\mathbb C^{2^n}) \rightarrow \mathcal{L}(\mathbb C^{2^n})$ be a single-qubit quantum channel such that $\tr(\channel \circ \mathrm{diag}) \neq 1$. Let $\hat{U} = \hat{U}_1 \otimes \ldots \otimes \hat{U}_n \in \ensemble$ and $\hat{b} = \hat{b}_1 \ldots \hat{b}_n \in \{0, 1\}^n$. Then, 
\begin{align}
    \hat{\rho}(\ensemble, \channel^{\otimes n}, \hat{U}, \hat{b}) = \bigotimes_{i = 1}^n \Big( \frac{1}{f(\channel)} \hat{U} \lvert \hat{b}\rangle\!\langle \hat{b}\rvert \hat{U} + \big( 1 - \frac{1}{f(\channel)}\big) \frac{\mathbb I}{2} \Big). 
\end{align}
\end{claim}
\begin{proof}
Follows from \cref{claim:global_clifford_classical_shadow} by setting $n = 1$.
\end{proof}

\subsection{Derivation of Shadow Seminorm}\label{subsection:product_clifford_shadow_norm}
We derive an expression for the shadow seminorm when the observable is a $k$-local Pauli observable, which non-trivially generalizes the Lemma S3 in \cite{huang2020predicting}.
Denote the set of Pauli operators by $\mathcal P_n = \{ P_1 \otimes \ldots \otimes P_n : P_i \in \{
\mathbb I, X, Y, Z\}\, \forall i \in \{1, \ldots, n\}\}$. 

\begin{definition}[weight of Pauli operator]
The \textit{weight} of the Pauli operator $P = P_1 \otimes \ldots \otimes P_n \in \mathcal P_n$ is $\mathrm{wt}(P) = \lvert \{ i : P_i \neq \mathbb I \}\rvert.$
\end{definition}

\begin{proposition}\label{proposition:pauli_shadow_norm}
Let $\ensemble = \ensemble_1 \otimes \ldots \otimes \ensemble_n$ be a product ensemble such that $\ensemble_i$ is a 3-design for all $i \in [n]$. 
Let $\channel: \mathcal{L}(\mathbb C^{2^n}) \rightarrow \mathcal{L}(\mathbb C^{2^n})$ be a single-qubit quantum channel such that $\tr(\channel \circ \mathrm{diag}) \neq 1$. Let $P \in \mathcal P_n$. Then, 

\begin{align}
    \norm{P}_{\mathrm{shadow},\ensemble, \channel^{\otimes n}} = \Big( \frac{1}{\sqrt{3} f(\channel)}\Big)^{\mathrm{wt}(P)}. 
\end{align}
where $f(\mathcal E) = \frac{1}{3}(\tr(\mathcal E\circ \mathrm{diag})-1)$. 
\end{proposition}
\begin{proof}
Without loss of generality, write $P = P_1 \otimes \ldots \otimes P_k \otimes \mathbb I^{\otimes (n-k)}$. Then, 
\begin{align}
    \norm{P}_{\mathrm{shadow},\ensemble, \channel^{\otimes n}}^2 
    &= \norm{P_1 \otimes \ldots \otimes P_k \otimes \mathbb I^{\otimes (n-k)}}_{\mathrm{shadow},\ensemble, \channel^{\otimes n}}^2 \nn
    &= \norm{P_1 \otimes \ldots \otimes P_k}_{\mathrm{shadow},\ensemble, \channel^{\otimes n}}^2 \nn
    &= \max_{\sigma \in \mathbb{D}_{2^n}} \underset{U \sim \localclifford^{\otimes k}}{\E} \sum_{b \in \{0,1\}^k} \langle b \rvert \channel^{\otimes k}(U \sigma U^\dagger) \lvert b \rangle\!\langle b \rvert U \bigotimes_{i=1}^k \mathcal{D}_{1, f(\channel)^{-1}}(P_i) U^\dagger \lvert b\rangle^2 \nn
    &= \max_{\sigma \in \mathbb{D}_{2^n}} \underset{U \sim \localclifford^{\otimes k}}{\E} \sum_{b \in \{0,1\}^k} \langle b \rvert \channel^{\otimes k}(U \sigma U^\dagger) \lvert b \rangle\!\langle b \rvert U \bigotimes_{i=1}^k\Big(\frac{1}{f(\channel)}P_i\Big) U^\dagger \lvert b\rangle^2 \nn
    &= \frac{1}{f(\channel)^{2k}} \max_{\sigma \in \mathbb{D}_{2^n}} \underset{U \sim \localclifford^{\otimes k}}{\E} \sum_{b \in \{0,1\}^k} \langle b \rvert \channel^{\otimes k}(U \sigma U^\dagger) \lvert b \rangle\!\langle b \rvert U \bigotimes_{i=1}^k P_i U^\dagger \lvert b\rangle^2.
\end{align}
The second equality follows from \cref{lemma:general_locality_respecting_shadow_norm}. 
To simplify the expression further, we write $\sigma = \sum_{\alpha, \beta \in \{0,1\}^k} \sigma_{\alpha \beta} \lvert \alpha \rangle\!\langle \beta \rvert = \sum_{\alpha, \beta \in \{0,1\}^k} \sigma_{\alpha \beta} E_{\alpha_1\beta_1} \otimes \ldots \otimes E_{\alpha_k\beta_k}$, where $E_{\alpha_i \beta_i} = \lvert \alpha_i \rangle\!\langle \beta_i \rvert$. 
We simplify the expectation value first.  

\begin{align}
    \underset{U \sim \localclifford^{\otimes k}}{\E} &\sum_{b \in \{0,1\}^k} \langle b \rvert \channel^{\otimes k}(U \sigma U^\dagger) \lvert b \rangle\!\langle b \rvert U \bigotimes_{i=1}^k P_i U^\dagger \lvert b\rangle^2 \nn
    &= \underset{U \sim \localclifford^{\otimes k}}{\E} \sum_{b \in \{0,1\}^k} \langle b \rvert \channel^{\otimes k}(U (\sum_{\alpha, \beta \in \{0,1\}^k} \sigma_{\alpha \beta}E_{\alpha_1 \beta_1} \otimes \ldots \otimes E_{\alpha_k \beta_k}) U^\dagger) \lvert b \rangle\!\langle b \rvert U \bigotimes_{i=1}^k P_i U^\dagger \lvert b\rangle^2 \nn
    &= \sum_{\alpha, \beta \in \{0,1\}^k} \sigma_{\alpha \beta} \prod_{i=1}^k \underset{U_j \sim \localclifford}{\E} \sum_{b_j \in \{0,1\}} \langle b_j \rvert \channel(U_j E_{\alpha_j \beta_j}  U^\dagger) \lvert b_j \rangle\!\langle b_j \rvert U_j P_j U_j^\dagger \lvert b_j\rangle^2 \nn
    &= \sum_{\alpha, \beta \in \{0,1\}^k} \sigma_{\alpha \beta} \prod_{i=1}^k \frac{1}{4!} \Big( 2(6 -2\tr(\channel \circ \mathrm{diag})) + 2(2\tr(\channel \circ \mathrm{diag})-2) \Big) \tr(E_{\alpha_j\beta_j}) \nn 
    &= \sum_{\alpha, \beta \in \{0,1\}^k} \sigma_{\alpha \beta} \prod_{i=1}^k \frac{1}{3}\delta_{\alpha_j \beta_j} \nn
    &= \frac{1}{3^k}.
\end{align}
The third equality follows from \cref{lem:identity3designsum} with $B = C = P_j$, $d=2$ and $A = E_{\alpha_j\beta_j}$. Plugging into the original expression, we get  
\begin{align}
    \norm{P}_{\mathrm{shadow},\ensemble, \channel^{\otimes n}}^2 
    &= \frac{1}{f(\channel)^{2k}} \max_{\sigma \in \mathbb{D}_{2^n}} \underset{U \sim \localclifford^{\otimes k}}{\E} \sum_{b \in \{0,1\}^k} \langle b \rvert \channel^{\otimes k}(U \sigma U^\dagger) \lvert b \rangle\!\langle b \rvert U \bigotimes_{i=1}^k P_i U^\dagger \lvert b\rangle^2 \nn
    &= \Big( \frac{1}{\sqrt{3} f(\channel)}\Big)^{2k} \\
    \implies \norm{P}_{\mathrm{shadow},\ensemble, \channel^{\otimes n}} &= \Big( \frac{1}{\sqrt{3} f(\channel)}\Big)^{\mathrm{wt}(P)}.
\end{align}
\end{proof}

Since we have a bound on the shadow seminorm, we can bound the sample complexity of classical shadows protocol when the unitary ensemble is the product Clifford group and when all the observables are $k$-local Pauli operators. 

\begin{corollary}\label{corollary:product_clifford_sample_complexity}
Let $\{P_i\}_{i=1}^M$ be a collection of $M$ Pauli operators. 
Let $\ensemble = \ensemble_1 \otimes \ldots \otimes \ensemble_n$ be a product ensemble such that $\ensemble_i^\dagger$ is a 3-design for all $i \in [n]$. 
Let $\channel: \mathcal{L}(\mathbb C^{2}) \rightarrow \mathcal{L}(\mathbb C^{2})$ be a single-qubit quantum channel such that $\tr(\channel \circ \mathrm{diag}) \neq 1$. 
The sample complexity $N_\mathrm{tot}$ to estimate the linear target functions $\{\tr(P_i \rho)\}_{i=1}^M$  of an $n$-qubit state $\rho$ within error $\eps$ and failure probability $\delta$ when unitary ensemble $\ensemble$ is subject to quantum channel $\channel^{\otimes n}$ is  

\[
N_\mathrm{tot} \leq   \frac{68\log(2M/\delta)}{\eps^2} \max_{1 \leq i \leq M} \Big( \frac{1}{3 f(\channel)^2}\Big)^{\mathrm{wt}(P_i)}, 
\]
where $f(\channel) = \frac{1}{3}(\tr(\channel \circ \mathrm{diag}) - 1)$ and $\mathrm{wt}(P) = \lvert \{ i : P_i \neq \mathbb I \}\rvert.$
\end{corollary}

\begin{remark}
We also studied the shadow seminorm of a general $k$-local observable in the presence of a general product channel. 
However, we could not derive a simple expression. 
In \cref{section:shadow_norm_klocal_observable_product}, we derive a bound on the shadow seminorm of a $k$-local observable when the error channel is depolarizing noise (rather than a general product channel). 
\end{remark}

\subsection{Examples}\label{subsec:local-examples}
We apply the results of this section to derive expressions for the shadow channel, inverse shadow channel, and the classical shadow when the unitary ensemble is the product Clifford ensemble and we fix the quantum channel. Specifically, we study the identity channel, depolarizing channel, and amplitude damping channel. 
For each quantum channel, we also bound the shadow seminorm for $k$-local Pauli observables, which imply sample complexity bounds for the classical shadows protocol with the product Clifford ensemble in the presence of noise. 

\subsubsection{Noiseless Channel}
We start with the noiseless case to show that our results can be used to recover Eq. (S44) through (S50) in \cite{huang2020predicting}.
Recall from \cref{subsec:global_noiseless_case}, we show $f(\mathbb I) = 1/(2^n+1)$. 
Hence, for the local case ($n = 1$), $f(\mathbb I) = 1/3$. 
It follows from \cref{claim:shadow_channel_product_clifford}  and  \cref{claim:classical_shadow_product_clifford}, that 
\[
\mathcal{M}_{\localclifford^{\otimes n}, \mathbb I} = \mathcal D_{1,1/3}^{\otimes n}
,\qquad
\mathcal{M}_{\localclifford^{\otimes n}, \mathbb I}^{-1} = \mathcal D_{1,3}^{\otimes n}
,\qquad \text{and} \qquad 
\hat{\rho}(\localclifford^{\otimes n}, \mathbb I, \hat{U}, \hat{b}) = \bigotimes_{i=1}^k (3 \hat{U} \lvert \hat{b_i}\rangle\!\langle \hat{b_i}\rvert \hat{U} - \mathbb I).
\]

Similarly, by using $f(\mathbb I) = 1/3$, the shadow seminorm of a $k$-local Pauli operator can be computed from \cref{proposition:pauli_shadow_norm}. 
\[
\norm{P}_{\mathrm{shadow}, \localclifford^{\otimes n}, \mathbb I}^2 = 3^{\mathrm{wt}(P)}.
\]

Finally, the sample complexity in this setting follows from \cref{corollary:product_clifford_sample_complexity}. 
\[
N_\mathrm{tot} \leq   \frac{68\log(2M/\delta)}{\eps^2} \max_{1 \leq i \leq M} 3^{\mathrm{wt}(P_i)}. 
\]

\subsubsection{Depolarizing Channel}
\label{sec:dep_channel}
Now we study the case where the quantum channel is $\mathcal{D}_{1, f}^{\otimes n}$ (see \cref{def:depolarizing-channel}). In \cref{subsec:global_depolarizing}, we showed that $f(\mathcal{D}_{n, f}) = f/(2^n + 1)$, and so, for $n=1$, we get $f(\mathcal D_{1,f}) = f/3$. 
By applying \cref{claim:shadow_channel_product_clifford}  and  \cref{claim:classical_shadow_product_clifford}, we get the following expressions for the shadow channel, inverse shadow channel, and classical shadow. 
\[
\mathcal{M}_{\localclifford^{\otimes n}, \mathcal{D}_{1,f}^{\otimes n}} = \mathcal D_{1,f/3}^{\otimes n}
,\,\,\,
\mathcal{M}_{\localclifford^{\otimes n}, \mathcal{D}_{1,f}^{\otimes n}}^{-1} = \mathcal D_{1,3/f}^{\otimes n}
,\,\,\, \text{and}
\]
\[ 
\hat{\rho}(\localclifford^{\otimes n}, \mathcal{D}_{1,f}^{\otimes n}, \hat{U}, \hat{b}) = \bigotimes_{i=1}^k \left(\frac{3}{f} \hat{U} \lvert \hat{b_i}\rangle\!\langle \hat{b_i}\rvert \hat{U} - \left(\frac{1}{2} - \frac{3}{2f}\right) \mathbb I\right).
\]

Applying \cref{proposition:pauli_shadow_norm} with  $f(\mathcal D_{1,f}) = f/3$, we get 
\[
\norm{P}_{\mathrm{shadow}, \localclifford^{\otimes n}, \mathcal D_{1,f}^{\otimes n}}^2 = \left( \frac{3}{f^2}\right)^{\mathrm{wt}(P)}.
\]

The sample complexity in this setting follows from \cref{corollary:product_clifford_sample_complexity}.
\[
N_\mathrm{tot} \leq   \frac{68\log(2M/\delta)}{\eps^2} \max_{1 \leq i \leq M}\left( \frac{3}{f^2}\right)^{\mathrm{wt}(P)}.
\]

\subsubsection{Amplitude Damping Channel}
\label{sec:ampli_damping}
The last example we consider is the case where the quantum channel is a product of local amplitude damping channels, denoted by $\mathrm{AD}_{1, p}^{\otimes n}$ (see \cref{def:amplitude-damping-channel}). In \cref{subsec:global_amplitude_damping}, we show that $f(\mathrm{AD}_{n, p}) =  \frac{(1+p)^n-1}{2^{2n}-1}$. For $n = 1$, $f(\mathrm{AD}_{1,p}) = p/3$.
We get expressions for the shadow channel and inverse shadow channel by applying \cref{claim:shadow_channel_product_clifford}.
\[
\mathcal{M}_{\localclifford^{\otimes n}, \mathrm{AD}_{1,p}^{\otimes n}} = \mathcal D_{1,p/3}^{\otimes n}
\qquad\text{and}\qquad
\mathcal{M}_{\localclifford^{\otimes n}, \mathrm{AD}_{1,p}^{\otimes n}}^{-1} = \mathcal D_{1,3/p}^{\otimes n}.
\]
Similarly, we apply \cref{claim:classical_shadow_product_clifford} to get an expression for the classical shadow. 
\[
\hat{\rho}(\localclifford^{\otimes n}, \mathrm{AD}_{1,p}^{\otimes n}, \hat{U}, \hat{b}) = \bigotimes_{i=1}^k \left(\frac{3}{p} \hat{U} \lvert \hat{b_i}\rangle\!\langle \hat{b_i}\rvert \hat{U} - \left(\frac{1}{2} - \frac{3}{2p}\right) \mathbb I\right).
\]

Applying \cref{proposition:pauli_shadow_norm} with  $f(\mathrm {AD}_{1,p}) = p/3$, we get 
\[
\norm{P}_{\mathrm{shadow}, \localclifford^{\otimes n}, \mathcal D_{1,f}^{\otimes n}}^2 = \left( \frac{3}{p^2}\right)^{\mathrm{wt}(P)}.
\]

The sample complexity in this setting follows from \cref{corollary:product_clifford_sample_complexity}.
\[
N_\mathrm{tot} \leq   \frac{68\log(2M/\delta)}{\eps^2} \max_{1 \leq i \leq M}\left( \frac{3}{p^2}\right)^{\mathrm{wt}(P)}.
\]

\section{Concluding Remarks and Open Problems}
\label{sec:conclusion}

In this paper, we generalized the Huang-Kueng-Preskill classical shadows protocol \cite{huang2020predicting} to take into account the effects of noise. We studied scenarios in which the quantum computer implementing the classical shadows protocol is subject to various noise channels. The noise models we considered include depolarizing noise, dephasing noise and amplitude damping noise.

For each of these noise models, we derived upper bounds for the number of samples needed to achieve expecatation value estimates with a given accuracy. These upper bounds are specified in terms of a shadow seminorm that we introduce in this paper. The shadow seminorm generalizes the shadow norm used to bound the sample complexity in the noiseless classical shadows protocol \cite{huang2020predicting}. By modifying the classical post-processing step of the noiseless protocol, we introduced a new estimator that remains unbiased in the presence of noise. A high-level takeaway of our work is that the classical shadows protocol is still efficient for certain estimation tasks, even in the presence of noise.

We conclude by listing a few open questions and future directions that could build on this work.

\begin{enumerate}
    \item 
    Comparison of our work with \cite{chen2020robust}.
    What if the true noise channel is given by $\mathcal E$, but the user thinks that the noise channel is given by $\mathcal F \neq \mathcal E$? This will result in the application of the inverse shadow channel $\mathcal M^{-1}_{\mathcal U,\mathcal F}$ instead of $\mathcal M^{-1}_{\mathcal U,\mathcal E}$, which will likely lead to a classical shadow that is not an unbiased estimator of $\rho$. We leave it as an open problem to analyze this setting and give bounds on the bias of the estimator. Furthermore, how does our work compare with the approach given in \cite{chen2020robust}?
    If $\mathcal F$ and $\mathcal E$ are ``close'' enough, is our approach preferred to \cite{chen2020robust}? 
    \item 
    Scope and limitations of our noise model.
    As noted in \cref{sec:our_contributions}, the assumption on our noise model is sometimes referred to as the \textit{GTM noise assumption}\footnote{\label{footnote:GTM_noise} For the GTM noise assumption, there is some freedom involved in whether the Markovian noise acts before or after the perfect application of the unitary operation. Indeed, some references (like \cite{flammia2020efficient,chen2022quantum}) have chosen to put the noise before the unitary and others (like \cite{chen2020robust} and this manuscript) have chosen to put the noise after the unitary. More generally, one could consider the case where known noise channels act both before and after the unitary; in \cref{sec:noisy_input_states}, we discuss this case and show that this leads to only a minor modification of the quantities involved in \cref{algo:full_protocol}.}, which allows the noisy channel $\tilde {\mathcal U}$ to be written as $\tilde {\mathcal U} = \mathcal N \circ \mathcal U$, where $\mathcal U = U(\cdot)U^\dag$ is the ideal unitary channel and $\mathcal N$ is a quantum channel that is independent of both $U$ and of the physical time at which the computation is performed.
    While this assumption is common in the literature (see \cite{chen2020robust,flammia2020efficient,chen2022quantum} and references therein), 
    it is likely too simplistic to represent noise on real devices \cite{young2020diagnosing}. 
    If one can give empirical evidence that our algorithm achieves higher accuracy than the original classical shadows protocol, then that would serve as evidence that the GTM assumption is not too simplistic for classical shadows.
    However, it is still possible that a more realistic noise model could lead to an algorithm that produces higher-accuracy estimates.
    We leave it as an open problem to analyze classical shadows with more realistic noise models, e.g., noise that is gate-dependent and/or non-Markovian.
    \item Experimental demonstration of the classical shadows protocol, where our work is used to improve the accuracy of the estimation. Specifically, this demonstration would involve running experiments for which the estimation accuracy is improved by inverting the noisy shadow channel (rather than the noiseless shadow channel). This would build on some experimental work on classical shadows that have been performed recently, for example, \cite{struchalin2021experimental,zhang2021experimental,liu2022experimental}. 
    \item Invertibility of the noisy shadow channel. Are there nice and simple necessary and sufficient conditions for invertibility of the noisy shadow channel? This would generalize the result stated in \cref{sec:classical_shadows_review} that a sufficient condition for the noiseless shadow channel to be invertible is that the unitary ensemble is tomographically complete; and would generalize Claim
    \ref{claim:invertibility_global_clifford}, which states that if the unitary ensemble is a 2-design and the noise channel is denoted by $\mathcal E$, then $\tr(\channel \circ \mathrm{diag}) \neq 1$ if and only if the shadow channel is invertible.
    \item Comparison of classical shadows with competing methods for estimating properties of quantum states on noisy quantum devices. When should one use classical shadows over other methods? This can be investigated theoretically, where one establishes theoretical performance guarantees for these methods; or numerically or experimentally, where one performs empirical comparisons between the performance of different methods.
    
    Examples of competing methods include those that we mentioned in \cref{sec:propertyEstimationLiterature} and \cref{sec:solvingMeasurementProblem}. We note here that there has been some recent work along this direction. 
    For example,
    recent work by Hadfield et al.\ have compared a non-uniform version of classical shadows (called locally-biased classical shadows) with competing methods like grouping and $\ell^1$-sampling, and have shown that it outperforms these other methods for the task of estimating expectation values of molecular Hamiltonians \cite{hadfield2022measurements}. An important next step would be to investigate if these advantages continue to hold in the presence of noise.
\end{enumerate}

\section*{Acknowledgements}

We thank Peter D.~Johnson for useful discussions. We thank Juan Carrasquilla and the anonymous reviewers for their valuable feedback on this manuscript. DEK acknowledges funding support from the National Research Foundation, Singapore, through Grant NRF2021-QEP2-02-P03.

\appendix

\section{Deferred Proofs}

\subsection{Proof of \texorpdfstring{\cref{lem:identity3designsum}}{Lemma}}
\label{app:identity3designsum}

To prove \cref{lem:identity3designsum}, we first introduce some notation. For a qudit linear superoperator $\Lambda:\mathcal L(\mathbb C^d) \rightarrow 
\mathcal L(\mathbb C^d)$ and $x \in [d]$, define
\begin{align}
    t_{\Lambda,x}&\overset{\mathrm{def}}{=} \tr(\Lambda(\ketbra xx)), \\
    \Lambda_{xxxx} &\overset{\mathrm{def}}{=} \bra x \Lambda(\ketbra xx) \ket x.
\end{align}
These functions satisfy the following properties:
\begin{claim}
\label{claim:propertiesddagger}
Consider the unary operator $()^\ddagger$ defined in
\cref{def:ddagger}. Then,
\begin{enumerate}
    \item 
    \begin{align}
    \sum_{x \in [d]} t_{\Lambda,x} = \sum_{x \in [d]} t_{\Lambda^\ddagger,x} = \tr(\Lambda(I)).
    \end{align}
    If $\Lambda$ is trace-preserving or unital, then
    \begin{align}
    \sum_{x \in [d]} t_{\Lambda,x} = d.
    \end{align}
    \item 
    \begin{align}
        (\Lambda^\ddagger)_{xxxx} = \Lambda_{xxxx}.
    \end{align} 
    \item
    \begin{align}
        \sum_{x\in [d]} \Lambda_{xxxx}
        =
        \sum_{x\in [d]}
        (\Lambda^\ddagger)_{xxxx}
        = \Tr(\Lambda\circ\diag).
    \end{align}
\end{enumerate}
\end{claim}
\begin{proof}
By straightforward calculation.
\end{proof}

We also need to evaluate the Haar integral in \cref{eq:t-foldTwirlHaar} when $t=3$. 
To evaluate the Haar integral, 
we first introduce some notation. Let $S_3 = \{1,(12),(13),(23),(123),(132)\}$ denote the symmetric group on three elements (where we have written its elements in cycle notation). For each $\pi \in S_3$, define the permutation operator $W_\pi \in \mathbb U((\mathbb C^d)^{\otimes n})$ to be the unique linear operator satisfying
\begin{align}
    W_\pi(x_1\otimes x_2 \otimes x_3) = x_{\pi^{-1}(1)} \otimes x_{\pi^{-1}(2)} \otimes x_{\pi^{-1}(3)}
\end{align}
for all $x_1,x_2, x_3 \in \mathbb C^3$. Equivalently,
\begin{align}
    W_\pi = \sum_{x \in \mathbb Z_d^n} \ket x \bra{\pi(x)},
\end{align}
where $\ket{\pi(x)}= \ket{x_{\pi(1)},\ldots,x_{\pi(n)}}$.

Following \cite{eggeling2001separability}, define the following linear combinations of permutation operators:
\begin{align}
\label{eq:R+}
    R_+ &= \frac 16 \sum_{\pi \in S_3} W_\pi = \frac 16 (I +W_{12}+W_{13} + W_{23} +W_{123}+W_{132}), \\
    R_- &= \frac 16 \sum_{\pi \in S_3} \mathrm{sgn}(\pi) W_\pi = \frac 16 (I - W_{12} - W_{13} - W_{23} +W_{123}+W_{132}), \\
    R_0 &= \frac 13(2I-W_{123}-W_{132}), \\
    R_1 &= \frac 13(2 W_{23}-W_{13}-W_{12}), \\
    R_2 &= \frac{1}{\sqrt 3}(W_{12}-W_{13}), \\
    \label{eq:R3}
    R_3 &= \frac{\ii}{\sqrt{3}}(W_{123}-W_{132}),
\end{align}
where we have dropped the parentheses in the notation for the permutation operators: $W_{12}=W_{(12)}, W_{13}=W_{(13)}$, etc.

When $t=3$, the Haar integral in \cref{eq:t-foldTwirlHaar} may be expressed in terms of the operators $R_i$ as follows \cite[Eq.~(A3)]{johnson2013compatible}:
\begin{align}
    \label{eq:HaarIntegralt3}
    T_3^{(d)}(A)
&=
\frac{6 \tr(R_+ A)}{d(d+1)(d+2)} R_+
+
\frac{6 \tr(R_- A)}{d(d-1)(d-2)} R_-
+
\frac{3}{2d(d^2-1)} \sum_{i=0}^3 \tr(R_i A) R_i.
\end{align}

We now state and prove the following identities.
\begin{lemma}
    Let $d\in \mathbb Z^+$ and $x \in \mathbb Z_d$. Let $\Lambda: \mathcal L(\mathbb C^d)\rightarrow \mathcal L(\mathbb C^d)$ be a linear superoperator, and let $A,B,\Gamma \in \mathcal L (\mathbb C^d)$ be linear operators. Then,
 \begin{align}
 \label{eq:T3dFirst}
        T_3^{(d)} (\ketbra{zxx}{yxx})
        &=
        \frac{2}{d(d+1)(d+2)}(\delta_{yz}+2\delta_{xyz}) R_+
        +
        \frac{1}{d(d+1)(d-1)} (\delta_{yz} - \delta_{xyz})(R_0+R_1).
        \\
        \label{eq:T3dSecond}
        T_3^{(d)}(\Gamma \otimes \ketbra{xx}{xx})
        &=
       \frac{2}{d(d+1)(d+2)}(
       \tr\Gamma 
       +
       2\bra x \Gamma \ket x
       ) R_+ \nn
        &\qquad +
       \frac{1}{d(d+1)(d-1)} (
       \tr\Gamma
       -
       \bra x \Gamma \ket x
       )(R_0+R_1).
       \\
       \label{eq:T3dThird}
    &\!\!\!\!\!\!\!\!\!\!\!\!\!\!\!\!\!\!\!\!\!\!\!\!\!\!\!\!\!\!\!\!\!\!\!\!\!\!\!\!\!\!\!\!\!
       \tr_{23}\{
       T_3^{(d)}(
       \Lambda(\ketbra xx) \otimes \ketbra{xx}{xx})(I\otimes B \otimes C)
       \} \nn
       &=
       \frac{1}{(d-1)d(d+1)(d+2)}\Big\{
       \left((1+d)t_{\Lambda,x} - 2\Lambda_{xxxx}\right)
       \left[\tr(BC) + \tr(B)\tr(C)\right]
       \nn
       &\quad\qquad\qquad + (d \Lambda_{xxxx}-t_{\Lambda,x}\left[
       B \tr(C) + C\tr(B) + BC + CB
       \right]
       \Big\}.
    \end{align}
    
\end{lemma}
\begin{proof} \hfill
\renewcommand\labelitemi{$\vcenter{\hbox{\tiny$\bullet$}}$}
\begin{itemize} 
    \item To prove \cref{eq:T3dFirst},
let $A= \ketbra{zxx}{yxx}$. Using the definitions of $R_i$ from Eqs.~\eqref{eq:R+}--\eqref{eq:R3}, we can calculate that
\begin{align}
    \tr(R_+ A) &= \frac 13(\delta_{yz} + 2 \delta_{xyz}), \\
    \tr(R_- A) &= \tr(R_2 A) = \tr(R_3 A) = 0, \\
    \tr(R_0 A) &= \tr(R_1 A)
    = \frac 23 (\delta_{yz}
    - \delta_{xyz}).
\end{align}
Substituting these into \cref{eq:HaarIntegralt3} gives \cref{eq:T3dFirst}.
\item
To prove \cref{eq:T3dSecond}, we decompose
\begin{align}
    \Gamma = \sum_{yz} \gamma_{yz} \ketbra zy.
\end{align}
Then, by linearity,
\begin{align}
    T_3^{(d)}(\Gamma \otimes \ketbra{xx}{xx})
    =
    \sum_{yz} \gamma_{yz} 
    T_3^{(d)} (\ketbra{zxx}{yxx}).
\end{align}

\cref{eq:T3dSecond} follows from \cref{eq:T3dFirst} and the facts that
\begin{align}
    \sum_{yz} \gamma_{yz} \delta_{yz} &= \tr\Gamma, \\
    \sum_{yz} \gamma_{yz} \delta_{xyz} &= \gamma_{xx} = \bra x \Gamma \ket x.
\end{align}
\item To prove \cref{eq:T3dThird}, set $\Gamma = \Lambda(\ketbra xx)$. Then, $\tr\Gamma = t_{\Lambda,x}$ and $\bra x \Gamma \ket x = \Lambda_{xxxx}$.

Using \cref{eq:T3dSecond},
\begin{align}
    T_3^{(d)}(
       \Lambda(\ketbra xx) \otimes \ketbra{xx}{xx})
    &=  \frac{2}{d(d+1)(d+2)}(
       t_{\Lambda,x}
       +
       2 \Lambda_{xxxx}
       ) R_+ \nn
        &\qquad +
       \frac{1}{d(d+1)(d-1)} (
       t_{\Lambda,x}
       -
       \Lambda_{xxxx}
       )(R_0+R_1).
\end{align}

Substituting this into the left-hand-side of \cref{eq:T3dThird} gives
\begin{align}
\label{eq:LHSexpression}
    \mathrm{LHS}
    =
    \frac{2(t_{\Lambda,x}+2\Lambda_{xxxx})}{d(d+1)(d+2)}
    \underbrace{\tr_{23}[R_+ (I\otimes B\otimes C)]}_{\circled{1}}
    +
    \frac{t_{\Lambda,x}-\Lambda_{xxxx}}{d(d+1)(d-1)}
    \underbrace{\tr_{23}[(R_0+R_1) (I\otimes B\otimes C)]}_{\circled{2}}.
\end{align}

Let
\begin{align}
    \xi_\pi = \tr_{23} (W_\pi (I \otimes B \otimes C)).
\end{align}
Expanding $R_+, R_0$ and $R_1$, we obtain
\begin{align}
\label{eq:Circled1LHS}
    \circled{1} &=
    \frac 16 \sum_{\pi \in S_3} \xi_\pi 
    \nn
    &= \tfrac 16
    (\tr(B)\tr(C) + \tr(BC) + B\tr(C) + C\tr(B) +BC + CB),
    \\
\label{eq:Circled2LHS}
    \circled{2} &=
    \tr_{23}\{
    [\tfrac 13 (2I - W_{123} - W_{132})
    +\tfrac 13 (2W_{23} - W_{13} - W_{12})
    ] (I\otimes B\otimes C)
    \} \nn
    &= \frac 13 (2\xi_1 - \xi_{(12)} - \xi_{(13)} + 2 \xi_{(23)} - \xi_{(123)} - \xi_{132})
    \nn
    &= \tfrac 13
    (2\tr(B)\tr(C) + 2\tr(BC) - B\tr(C) - C \tr(B) - BC - CB),
\end{align}
where we used the following identities
\begin{align}
    \xi_1 &= \tr(B) \tr(C) \\
    \xi_{(12)} &= B \tr(C) \\
\xi_{(13)} &= C \tr(B) \\
\xi_{(23)} &= \tr(BC) \\
\xi_{(123)} &= CB  \\
\xi_{(132)} &= B C.
\end{align}

Substituting \cref{eq:Circled1LHS} and \cref{eq:Circled2LHS} into \cref{eq:LHSexpression} and rearranging terms gives \cref{eq:T3dThird}.
\end{itemize}
\end{proof}

We are now ready to prove
\cref{lem:identity3designsum}. Let $\Udag$ be a qudit 3-design and $\mathcal E:\mathcal L(\mathbb C^d) \rightarrow \mathcal L(\mathbb C^d)$ be a linear superoperator. Let $b \in [d]$ and let $A,B,C \in \mathcal L(\mathbb C^d)$ be linear operators. First, we consider the function
\begin{align}
    \Xi_{\mathcal E}(b):&= \E\limits_{U\sim \mathcal U} U^\dag \mathcal E(\ketbra bb) U \bra b UBU^\dag \ket b \bra b UCU^\dag \ket b \nn
    &=
    \E\limits_{U\sim \mathcal U} U^\dag \mathcal E(\ketbra bb) U \bra {bb} (U\otimes U) (B\otimes C)(U^\dag \otimes U^\dag) \ket{bb}
    \nn
    &=
    \tr_{23} \E_{U\sim \mathcal U}
    \left\{
    U^\dag \mathcal E(\ketbra bb) U \otimes (U^\dag \otimes U^\dag) 
    \ketbra{bb}{bb}(U\otimes U)
    (B \otimes C)
    \right\}\nn
    &=
    \tr_{23}\left[
    \E\limits_{U\in \mathcal U}
    (U^\dag \otimes U^\dag \otimes U^\dag) \mathcal E(\ketbra bb)\otimes \ketbra{bb}{bb} (U\otimes U\otimes U)(I\otimes B \otimes C)\right] \nn
    &=
    \tr_{23}\left[
    \E\limits_{U\in \mathcal U^\dag}
    (U \otimes U \otimes U) \mathcal E(\ketbra bb)\otimes \ketbra{bb}{bb} (U^\dag\otimes U^\dag \otimes U^\dag)(I\otimes B \otimes C)
    \right] \nn
    &= \tr_{23}\left[
    T_3^{(d)}(\mathcal E(\ketbra bb)\otimes \ketbra{bb}{bb}) (I \otimes B \otimes C
    )\right] \nn
    &= \frac{1}{(d-1)d(d+1)(d+2)}
    \Big\{
    \left[
    (1+d) t_{\mathcal E,b} - 2\mathcal E_{bbbb}
    \right]\left[
    \tr(BC)+\tr(B)\tr(C)
    \right] \nn
    &\quad +
    (d \mathcal E_{bbbb}-t_{\mathcal E,b})
    \left[
    B \tr(C) + C\tr(B) + BC + CB
    \right] \Big\}
    \label{eq:generalizedEqS36}
\end{align}
where the sixth line follows from the assumption that $\mathcal U^\dag = \{U:U^\dag \in \mathcal U\}$ is a 3-design, and the last line follows from
\cref{eq:T3dThird}.

Summing $\Xi_{\mathcal E} (b)$ over all $b$,
we obtain
\begin{align}
    \sum_{b\in [d]} \Xi_{\mathcal E} (b)
    &=
    \sum_{b\in [d]} \Xi_{\mathcal E^\ddagger} (b) \nn
    &=
    \label{eq:sumXib}
    \frac{1}{(d-1)d(d+1)(d+2)}
    \Big\{
    \left[
    (1+d) \tr(\mathcal E(I)) - 2\tr(\mathcal E\circ \diag)
    \right]\left[
    \tr(BC)+\tr(B)\tr(C)
    \right] \nn
    &\quad +
    (d \tr(\mathcal E\circ \diag)
    -\tr(\mathcal E(I))
    \left[
    B \tr(C) + C\tr(B) + BC + CB
    \right]
    \Big\},
\end{align}
where $\mathcal E^\ddagger(A) = (\mathcal E^*(A^\dag))^\dag$. We use the (easy-to-verify) fact that $\mathcal E$ is trace-preserving iff 
$\mathcal E^\ddagger$ is unital, which implies that 
$\tr(\mathcal E(I)) = \tr(\mathcal E^\ddagger(I))$ and 
$\Tr(\mathcal E \circ \diag) =  \Tr(\mathcal E^\ddagger \circ \diag)$.

Hence,
\begin{align}
    & \E\limits_{U\sim \mathcal U} \sum_{b\in[d]} 
    \bra b \mathcal E(UAU^\dag) \ket b \bra b UBU^\dag \ket b \bra b UCU^\dag \ket b \nn
    &\quad = 
    \E\limits_{U\sim \mathcal U} \sum_{b\in[d]} \tr\{\mathcal E(UAU^\dag) \ketbra bb\} \bra b UBU^\dag \ket b \bra b UCU^\dag \ket b 
    \nn
    &\quad =
    \E\limits_{U\sim \mathcal U} \sum_{b\in[d]} \tr\{U A U^\dag \mathcal E^\ddagger(\ketbra bb) \} \bra b UBU^\dag \ket b \bra b UCU^\dag \ket b 
    \nn
    &\quad =
    \E\limits_{U\sim \mathcal U} \sum_{b\in[d]} \tr\{A U^\dag \mathcal E^\ddagger(\ketbra bb) U \} \bra b UBU^\dag \ket b \bra b UCU^\dag \ket b
    \nn
    &\quad =
    \tr\left\{ A
    \sum_{b\in[d]} \E\limits_{U\sim \mathcal U}
    U^\dag \mathcal E^\ddagger(\ketbra bb) U \bra b UBU^\dag \ket b \bra b UCU^\dag \ket b
    \right\}
    \nn
    &\quad = \tr\left\{A \sum_{b\in [d]} \Xi_{\mathcal E^\ddagger} (b)\right\}
    \nn
    &\quad = \frac{(1+d)\alpha-2\beta}{(d-1)d(d+1)(d+2)} 
    \left(
    \tr(A)\tr(BC) + \tr(A)\tr(B) \tr(C)
    \right) \nn
    &\qquad\quad+
    \frac{d\beta-\alpha}{(d-1)d(d+1)(d+2)}
    \left(\tr(AB) \tr(C)
    +
    \tr(AC) \tr(B)
    +
    \tr(ABC)
    +
    \tr(ACB)
    \right),
\end{align}
where $\alpha = \tr(\mathcal E(I))$ and $\beta = \Tr(\mathcal E \circ \diag)$, and where the last line follows from applying \cref{eq:sumXib}.

\subsection{Proof of \texorpdfstring{\cref{lemma:upper-bound-on-variance}}{Lemma}}\label{subsec:proof-of-lemma-upper-bound-on-variance}

Below we give the proof of \cref{lemma:upper-bound-on-variance}.

\begin{proof}
By definition, the variance of $\hat{o}$ is 

\[
\underset{\substack{U \sim \ensemble \\ b \sim P_{b}}}{\Var}[\hat{o}] 
= \underset{\substack{U \sim \ensemble \\ b \sim P_{b}}}{\E} \Big[ \big(\hat{o} - 
\underset{\substack{U \sim \ensemble \\ b \sim P_{b}}}{\E}[\hat{o}] \big)^2 \Big]. 
\]
Let $O_o = O - \tr(O)\frac{\mathbb{I}}{2^n}$ be the traceless part of $O$. Because $\tr(\rho) = \tr(\hat{\rho}) = 1$, it follows that $\hat{o} - \E[\hat{o}] = \tr(O_o \hat{\rho}) = \tr(O_o \rho)$. Therefore, the variance depends only on $O_o$: 

\begin{align}
\underset{\substack{U \sim \ensemble \\ b \sim P_{b}}}{\Var}[\hat{o}] 
= \underset{\substack{U \sim \ensemble \\ b \sim P_{b}}}{\E} \Big[ \big(\tr(O_o \hat{\rho}) - \tr(O_o \rho)\big)^2 \Big].
\end{align}
Simplifying further, we get

\begin{align}
\underset{\substack{U \sim \ensemble \\ b \sim P_{b}}}{\Var}[\hat{o}] 
&= \underset{\substack{U \sim \ensemble \\ b \sim P_{b}}}{\E} \Big[ \big(\tr(O_o \hat{\rho}) - \tr(O_o \rho)\big)^2 \Big] \nn 
&= \underset{\substack{U \sim \ensemble \\ b \sim P_{b}}}{\E} \Big[ \tr(O_o \hat{\rho})^2 + \tr(O_o \rho)^2 - 2 \tr(O_o \rho) \tr(O_o \hat{\rho}) \Big] \nn
&= \underset{\substack{U \sim \ensemble \\ b \sim P_{b}}}{\E} \Big[ \big(\tr(O_o \hat{\rho})^2\Big] + \tr(O_o \rho)^2 - 2 \tr(O_o \rho) \underset{\substack{U \sim \ensemble \\ b \sim P_{b}}}{\E} \Big[ \tr(O_o \hat{\rho}) \Big] \nn
&= \underset{\substack{U \sim \ensemble \\ b \sim P_{b}}}{\E} \Big[ \tr(O_o \hat{\rho})^2\Big] - \tr(O_o \rho)^2, 
\end{align}
where 
\begin{align}
    \tr(O_o \hat{\rho}) &= \tr\big(O_o \inverseshadowchannel{\ensemble}{\channel}(U^\dagger \lvert b \rangle\langle b \rvert U)\big) \nn
    &= \tr(\mathcal{M}^{-1, \dagger}_{\ensemble,\channel}(O_o) U^\dagger \lvert b \rangle\langle b \rvert U)\nn
    &= \langle b \rvert U \mathcal{M}^{-1, \dagger}_{\ensemble,\channel}(O_o)  U^\dagger \lvert b \rangle.
\end{align}

To get \textit{a priori} bounds on the variance, we must remove the dependence on the input state $\rho$, which we do by maximizing over all quantum states.  

\begin{align*}
\underset{\substack{U \sim \ensemble \\ b \sim P_{b}}}{\Var}[\hat{o}] 
&= \underset{\substack{U \sim \ensemble \\ b \sim P_{b}}}{\E} \Big[ \langle b \rvert U \mathcal{M}^{-1, \dagger}_{\ensemble, \channel}(O_o)  U^\dagger \lvert b \rangle^2 \Big] - \tr(O_o \rho)^2 \\ 
&= \underset{\substack{U \sim \ensemble }}{\E} \sum_{b \in \{0,1\}^n} P_b(b ; U, \channel, \rho) \langle b \rvert U \mathcal{M}^{-1, \dagger}_{\ensemble,\channel}(O_o)  U^\dagger \lvert b \rangle^2 - \tr(O_o \rho)^2 \\ 
&=  \underset{\substack{U \sim \ensemble }}{\E} \sum_{b \in \{0,1\}^n} \langle b \lvert \channel(U \rho U^\dagger) \lvert b \rangle\! \langle b \rvert U \mathcal{M}^{-1, \dagger}_{\ensemble,\channel}(O_o)  U^\dagger \lvert b \rangle^2 - \tr(O_o \rho)^2 \\ 
&\leq \max_{\sigma \in \mathbb{D}_{2^n}} \underset{\substack{U \sim \ensemble }}{\E} \sum_{b \in \{0,1\}^n} \langle b \lvert \channel(U \sigma U^\dagger) \lvert b \rangle\! \langle b \rvert U \mathcal{M}^{-1, \dagger}_{\ensemble,\channel}(O_o)  U^\dagger \lvert b \rangle^2 - \tr(O_o \rho)^2 \\ 
&= \bigg( \max_{\sigma \in \mathbb{D}_{2^n}}\sqrt{ \underset{\substack{U \sim \ensemble }}{\E} \sum_{b \in \{0,1\}^n} \langle b \rvert \channel(U \sigma U^\dagger) \lvert b \rangle\! \langle b \rvert U \mathcal{M}^{-1, \dagger}_{\ensemble,\channel}(O_o)  U^\dagger \lvert b \rangle^2}\bigg)^2 - \tr(O_o \rho)^2 \\ 
&=  \norm{O_o}_{\shadow, \ensemble, \channel}^2 - \tr(O_o \rho)^2 \\ 
&\leq  \norm{O - \tr(O)\frac{\mathbb{I}}{2^n}}_{\shadow, \ensemble, \channel}^2.
\end{align*}
\end{proof}

\subsection{Proof of \texorpdfstring{\cref{prop:global_shadow_norm_bounds}}{Corollary}}\label{app:proof-of-global-shadow-norm-bounds}

\begin{proof}
We use the fact that if $A \in \mathbb H_{2^n}$, then $2^{-n} \tr(A) \leq \norm{A}_\mathrm{sp} \leq \tr(A)$. First, we prove the lower bound. 

\begin{align}
    \norm{O_o^2}_\mathrm{sp} \geq \frac{1}{2^n}\tr(O_o^2) \implies \norm{O_o}^2_{\shadow, \ensemble, \channel} &\geq  \frac{2^{2n}-1}{(2^n + 2)(\beta - 1)} \Big( \frac{2^n + 2^{2n} - 2\beta}{2^n(\beta -1)} + \frac{2}{2^n} \Big) \tr(O_o^2)  \nn
    &= \frac{(2^{2n} - 1)(2^n -1)}{2^n(\beta - 1)^2}  \tr(O_o^2) \nn
    &= \frac{(2^{n} - 1)^2(2^n +1)}{2^n(\beta - 1)^2}  \tr(O_o^2) \nn 
    &\geq \frac{(2^{n} - 1)^2}{(\beta - 1)^2}  \tr(O_o^2).
\end{align}

Now, the first upper bound. 

\begin{align}
    \norm{O_o^2}_\mathrm{sp} \leq \tr(O_o^2) \implies \norm{O_o}^2_{\shadow, \ensemble, \channel} &\leq  \frac{2^{2n}-1}{(2^n + 2)(\beta - 1)} \Big( \frac{2^n + 2^{2n} - 2\beta}{2^n(\beta -1)} + 2 \Big) \tr(O_o^2)  \nn
    &= \frac{(2^n -1)^2(2^n +1)(2\beta + 2^n)}{2^n(2^n + 2) (\beta - 1)^2} \tr(O_o^2) \nn
    &\leq \frac{(2^n -1)^2(2\beta + 2^n)}{2^n(\beta - 1)^2} \tr(O_o^2) \nn
    &\leq \frac{(2^n -1)^2(2 \cdot 2^n + 2^n)}{2^n(\beta - 1)^2} \tr(O_o^2) \nn 
    &=\frac{3(2^n -1)^2}{(\beta - 1)^2} \tr(O_o^2). 
    \end{align}

Finally, the second upper bound.

\begin{align}
\norm{O_o}^2_{\shadow, \ensemble, \channel} 
&\leq \frac{3(2^n-1)^2}{(\beta -1)^2}\tr(O_o^2) \nn
&= \frac{3(2^n-1)^2}{(\beta -1)^2}\tr\big( (O - \frac{1}{2^n}\tr(O)\mathbb I)^2 \big) \nn
&= \frac{3(2^n-1)^2}{(\beta -1)^2}\big( \tr(O^2) - \frac{1}{2^n}\tr(O)^2 \big) \nn
&\leq \frac{3(2^n-1)^2}{(\beta -1)^2}\tr(O^2).
\end{align}
\end{proof}

\section{When is the Shadow Seminorm a Norm?}\label{app:when-is-shadow-norm}

In this appendix, we prove that the shadow seminorm is indeed a seminorm. In addition, we prove sufficient conditions for the shadow seminorm $\norm{\cdot}_{\mathrm{shadow}, \ensemble, \channel}$ to be a norm. The paradigmatic quantum channels studied in this work—namely the depolarizing channel, dephasing channel and amplitude damping channel—satisfy these conditions, and hence yield shadow seminorms which are also norms. 
We begin by stating the main result of this section. 

\begin{definition}
Let $\Lambda_n$ be the set of $n$-qubit quantum channels $\Theta$ satisfying the following:
\[
\forall b \in \{0,1\}^{n},\,\, \exists \,\sigma \in \mathbb{D}_{2^n} : \bra{b}\Theta(\sigma) \ket{b} \neq 0.
\]
\end{definition}

\begin{proposition}\label{prop:when-is-a-norm}
Let $\ensemble$ be an $n$-qubit unitary ensemble and $\channel$ an $n$-qubit quantum channel such that the shadow channel $\mathcal{M}_{\ensemble, \channel}$ is invertible. Then,
\begin{enumerate}
    \item The shadow seminorm $\norm{\cdot}_{\mathrm{shadow}, \ensemble, \channel}$ is a seminorm.
    \item If $\channel \in \Lambda_n$, then the shadow seminorm $\norm{\cdot}_{\mathrm{shadow}, \ensemble, \channel}$ is a norm.
\end{enumerate}
\end{proposition}

The rest of this appendix is dedicated to proving this statement (which amounts to explicitly verifying that the shadow seminorm satisfies the properties of being a seminorm or norm).

First, we show that the condition above doesn't trivially include all channels (i.e., there are channels that are not in $\Lambda_n$). We also show that all unital channels are in $\Lambda_n$.

\begin{claim}\hfill
\begin{enumerate}
    \item \label{item:seminorm1} There exist quantum channels which are not in $\Lambda_n$.
    \item \label{item:seminorm2} If a quantum channel $\Theta$ is unital, then $\Theta \in \Lambda_n$. 
\end{enumerate}
\end{claim}
\begin{proof}
To prove (\ref{item:seminorm1}), choose $\Theta$ to be the following map: $A \mapsto \tr(A)\ket{0_n}\!\bra{0_n}$ ($\ket{0_n}$ denotes the $n$-qubit state where all qubits are in the state $\ket{0}$). From the following Kraus representation of $\Theta$, we see that $\Theta$ is a quantum channel: 
\[
\Theta(A) = \tr(A) \ket{0_n}\!\bra{0_n} = \sum_i \bra{i}A\ket{i}\ket{0_n}\!\bra{0_n} = 
\sum_i \ket{0_n}\!\bra{i}A\ket{0_n}\!\bra{i}^\dagger.
\]
Now take $b = 1_n = 11\ldots1$ (that is, $\ket{b} = \ket{1_n}$ is the $n$-qubit state where all qubits are in the state $\ket{1}$). Then, $\forall \sigma \in \mathbb{D}_{2^n}, \bra{b}\Theta(\sigma)\ket{b} = 0$.
Thus, $\Theta \not\in \Lambda_n$.

To prove (\ref{item:seminorm2}), assume that the quantum channel $\Theta$ is unital. Let $b \in \{0,1\}^n$ and choose $\sigma$ to be the maximally mixed state. Then, 
\[
\bra{b}\Theta(\sigma)\ket{b} = \frac{1}{2^n} \bra{b}\Theta(I)\ket{b} = \frac{1}{2^n} \neq 0.
\]
\end{proof}

\begin{proof}[Proof of \cref{prop:when-is-a-norm}]
\hfill
\begin{enumerate}
    \item 
To show that the shadow seminorm is a seminorm, we shall explicitly verify that it satisfies the triangle inequality and absolute homogeneity.

First, the triangle inequality. 
\begin{align}
   \norm{S + T}_{\mathrm{shadow}, \ensemble, \channel} 
   = \max_{\sigma \in \mathbb{D}_{2^n}} \sqrt{\vphantom{\E \sum_{b \in \{0,1\}^n}} \smash[b]{\underbrace{\E_{U \sim \ensemble} \sum_{b \in \{0,1\}^n} \bra{b} \channel(U \sigma U^\dagger \ket{b} \bra{b}U \mathcal{M}_{\ensemble, \channel}^{-1, \dagger}(S+T) U^\dagger \ket{b}^2}_{\circled{1}}}}.
\end{align}
Then,
\begin{align}
   \circled{1} 
   &= \E_{U \sim \ensemble} \sum_{b \in \{0,1\}^n} \bra{b} \channel(U \sigma U^\dagger \ket{b} \bra{b}U \mathcal{M}_{\ensemble, \channel}^{-1, \dagger}(S+T) U^\dagger \ket{b}^2 \nn
   &= \E_{U \sim \ensemble} \sum_{b \in \{0,1\}^n} \bra{b} \channel(U \sigma U^\dagger \ket{b} \Big(\bra{b}U \mathcal{M}_{\ensemble, \channel}^{-1, \dagger}(S) U^\dagger \ket{b}^2 + \ket{b} \bra{b}U \mathcal{M}_{\ensemble, \channel}^{-1, \dagger}(T) U^\dagger \ket{b}^2\Big) \nn
   &= \E_{U \sim \ensemble} \sum_{b \in \{0,1\}^n} (\alpha_{b,U} + \beta_{b, U})^2 \nn
   &= \E_{U \sim \ensemble} \sum_{b \in \{0,1\}^n} \alpha_{b,U}^2 +
   \E_{U \sim \ensemble} \sum_{b \in \{0,1\}^n} 
   \beta_{b, U}^2 + 
   2 \E_{U \sim \ensemble} \sum_{b \in \{0,1\}^n} 
   \alpha_{b,U}\beta_{b, U},
\end{align}
where we define 
\[
\alpha_{b,U} \overset{\mathrm{def}}{=} \sqrt{\bra{b} \channel(U \sigma U^\dagger) \ket{b}} \bra{b}U \mathcal{M}_{\ensemble,\channel}^{-1, \dagger} (S) U^\dagger \ket{b},
\]
\[
\beta_{b,U}\overset{\mathrm{def}}{=} \sqrt{\bra{b} \channel(U \sigma U^\dagger) \ket{b}} \bra{b}U \mathcal{M}_{\ensemble,\channel}^{-1, \dagger} (T) U^\dagger \ket{b}.
\]
Then, by the Cauchy-Schwarz inequality, 
\begin{align}
   \circled{1}  
    &\leq \E_{U \sim \ensemble} \sum_{b \in \{0,1\}^n} \alpha_{b,U}^2 +
   \E_{U \sim \ensemble} \sum_{b \in \{0,1\}^n} 
   \beta_{b, U}^2 + 
   2 \sqrt{ \E_{U \sim \ensemble} \sum_{b \in \{0,1\}^n} 
   \alpha_{b,U}^2}\sqrt{  \E_{U \sim \ensemble} \sum_{b \in \{0,1\}^n} \beta_{b, U}^2} \nn
   &= 
  \left( \sqrt{ \E_{U \sim \ensemble} \sum_{b \in \{0,1\}^n} 
   \alpha_{b,U}^2} + \sqrt{  \E_{U \sim \ensemble} \sum_{b \in \{0,1\}^n} \beta_{b, U}^2}\right)^2.
\end{align}
Plugging this back into the original expression gives 
\begin{align}
   \norm{S + T}_{\mathrm{shadow}, \ensemble, \channel} 
   &= \max_{\sigma \in \mathbb{D}_{2^n}}     \sqrt{ \E_{U \sim \ensemble} \sum_{b \in \{0,1\}^n} \alpha_{b,U}^2} + \sqrt{  \E_{U \sim \ensemble} \sum_{b \in \{0,1\}^n} \beta_{b, U}^2} \nn
   &\leq \max_{\sigma \in \mathbb{D}_{2^n}}     \sqrt{ \E_{U \sim \ensemble} \sum_{b \in \{0,1\}^n} \alpha_{b,U}^2} + \max_{\sigma \in \mathbb{D}_{2^n}}  \sqrt{  \E_{U \sim \ensemble} \sum_{b \in \{0,1\}^n} \beta_{b, U}^2} \nn
   &= \norm{S}_{\mathrm{shadow}, \ensemble, \channel}  + \norm{T}_{\mathrm{shadow}, \ensemble, \channel}.
\end{align}
Secondly, observe that absolute homogeneity follows immediately by linearity.
\item 
We shall verify that if $\mathcal E \in \Lambda_n$ is satisfied, then the shadow seminorm is point-separating/positive semi-definite, which implies that it is also a norm.

Suppose that $\norm{T}_{\mathrm{shadow}, \ensemble, \channel} = 0$. Then
\begin{align}
   &\max_{\sigma \in \mathbb{D}_{2^n}} \sqrt{\E_{U \sim \ensemble} \sum_{b \in \{0,1\}^n} \bra{b} \channel(U \sigma U^\dagger \ket{b} \bra{b}U \mathcal{M}_{\ensemble, \channel}^{-1, \dagger}(T) U^\dagger \ket{b}^2} = 0 \nn
   &\implies  \bra{b} \channel(U \sigma U^\dagger \ket{b} \bra{b}U \mathcal{M}_{\ensemble, \channel}^{-1, \dagger}(T) U^\dagger \ket{b}^2 = 0, \quad \forall \sigma \in \mathbb{D}_{2^n}, U \in \ensemble, b \in \{0,1\}^n.
\end{align}
Without loss of generality, choose $U = \mathbb{I}$. Since $\channel \in \Lambda_n$, we know that $\forall \sigma \in \mathbb{D}_{2^n}$, there is a $b$ such that $\bra{b}\channel(\sigma) \ket{b} \neq 0$. 
Therefore, we can conclude that 
\[
  \bra{b} \mathcal{M}_{\ensemble, \channel}^{-1, \dagger}(T) \ket{b}^2 = 0, \quad \forall b \in \{0,1\}^n 
  \implies
  \mathcal{M}_{\ensemble, \channel}^{-1, \dagger}(T) = 0
  \implies T = 0,
\]
which completes the proof of our proposition.
\end{enumerate}
\end{proof}

\section{Inconsequential Noise}\label{app:inconsequential-noise}

Recall from \cref{lemma:2design} that when the unitary ensemble is a $2$-design, the shadow channel is a depolarizing channel with depolarizing parameter $f(\channel) = \frac{1}{2^{2n}-1} \Big( \tr(\channel \circ \mathrm{diag}) - \frac{1}{2^n} \tr(\channel (\mathbb I)) \Big)$.
In this setting, 
we characterize inconsequential noise (i.e., the quantum channels that do not affect the classical shadows protocol). 

\begin{claim}\label{claim:inconsequential_noise}
Let $\Udag$ be an $n$-qubit 2-design and let $\channel$ be a linear superoperator. $\channel$ has no effect on $\mathcal M_\ensemble$ (i.e., $\mathcal{M}_{\ensemble, \channel} = \mathcal{M}_\ensemble$) if and only if $\tr(\channel \circ \mathrm{diag}) = 2^n$.
Also, 
$\channel$ has no effect on $\mathcal M_\ensemble$ (i.e., $\mathcal{M}_{\ensemble, \channel} = \mathcal{M}_\ensemble$) if and only if $\langle b \rvert \channel(\lvert b \rangle\!\langle b \rvert) \lvert b \rangle = 1, \,\forall\, b \in \{0,1\}^n$. 
\end{claim}
\begin{proof}
\begin{align}
\mathcal{M}_{\ensemble, \channel} = \mathcal{M}_{\ensemble} \iff \mathcal{D}_{n, f(\channel)} = \mathcal{D}_{n, f(\mathbb{I})} \iff f(\channel) = f(\mathbb{I}) \iff \tr(\channel \circ \mathrm{diag}) = 2^n.
\end{align}
Also, 
\begin{align}
\mathcal{M}_{\ensemble, \channel} = \mathcal{M}_{\ensemble} \iff \tr(\channel \circ \mathrm{diag}) = 2^n \iff \langle b \rvert \channel(\lvert b \rangle\!\langle b \rvert) \lvert b \rangle = 1, \,\forall\, b \in \{0,1\}^n. 
\end{align}
\end{proof}

\section{Shadow Seminorm of \texorpdfstring{$k$}{k}-Local Observable with Product Clifford Ensemble}\label{section:shadow_norm_klocal_observable_product}
In this section we bound the shadow seminorm of a $k$-local observable when the product Clifford ensemble is subject to depolarizing noise (rather than a general quantum channel). 

\begin{proposition}
Let $O \in \mathbb H_{2^n}$ be a $k$-local observable with nontrivial part $\tilde{O}$. Let $\tilde{O} = \sum_{\mathbf{p} \in \mathbb Z_4^k} \alpha_\mathbf{p} P_\mathbf{p}$ be the expansion of $\tilde{O}$ in the Pauli basis. Let $0 \leq f \leq 1$. Then, 

\begin{align}
    \norm{O}_{\shadow, \localclifford^{\otimes n}, \mathcal{D}_{1, f}^{\otimes n}}^2 = \norm{\sum_{\mathbf p, \mathbf q \in \mathbb Z_4^k} \alpha_{\mathbf p} \alpha_{\mathbf q}\tilde{\mathcal{F}}(\mathbf p,\mathbf q)P_{\mathbf p}P_\mathbf q}_{\mathrm{sp}},
\end{align}
where 
\begin{align}
    \tilde{\mathcal{F}}(\mathbf p, \mathbf q) = \prod_{j=1}^k \tilde{f}(\mathbf p_j, \mathbf q_j),
\end{align}
with
\begin{align}
    \tilde{f}(p,q) = \begin{cases}
    1/f &\text{if $p=q=0$.} \\
    1 &\text{if $(p=0) \oplus (q=0)$.} \\
    3/f^2 &\text{if $p=q\neq 0$.} \\
    0 &\text{otherwise.}
    \end{cases}
\end{align}
\end{proposition}
\begin{proof}
Without loss of generality, we write $O = \tilde{O} \otimes \mathbb I$, where $\tilde{O} = \sum_{\mathbf p \in \mathbb Z_4^k} \alpha_\mathbf p P_\mathbf p$. Then, 

\begin{align}
    &\norm{O}^2_{\shadow, \localclifford^{\otimes n}, \mathcal D_{1,f}^{\otimes n}} \nn 
    &\quad= \norm{\tilde{O}}^2_{\shadow, \localclifford^{\otimes k}, \mathcal D_{1,f}^{\otimes k}}\nn
    &\quad= \max_{\sigma \in \mathbb D_{2^k}} \underset{U \sim \localclifford^{\otimes k}}{\E} \sum_{b \in \B^k} \langle b \rvert \mathcal D_{1,f}^{\otimes k}(U \sigma U^\dagger) \lvert b \rangle\!\langle b \rvert U \mathcal D_{1, 3/f}^{\otimes k} (\tilde{O}) U^\dagger \lvert b \rangle^2 \nn
    &\quad= \max_{\sigma \in \mathbb D_{2^k}} \underset{U \sim \localclifford^{\otimes k}}{\E} \sum_{b \in \B^k} \langle b \rvert \mathcal D_{1,f}^{\otimes k}(U \sigma U^\dagger) \lvert b \rangle\!\langle b \rvert U \mathcal D_{1, 3/f}^{\otimes k} \left(\sum_{\mathbf p \in \mathbb Z_4^k}\alpha_\mathbf p P_\mathbf p\right) U^\dagger \lvert b \rangle^2  \nn
    &\quad= \max_{\sigma \in \mathbb D_{2^k}}\! \sum_{\mathbf p, \mathbf q \in \mathbb Z_4^k}\!\!\! \alpha_\mathbf p \alpha_\mathbf q \underbrace{\underset{U \sim \localclifford^{\otimes k}}{\E} \sum_{b \in \B^k}\!\!\! \langle b \rvert \mathcal D_{1,f}^{\otimes k}(U \sigma U^\dagger) \lvert b \rangle\!\langle b \rvert U \mathcal D_{1, 3/f}^{\otimes k} (P_\mathbf p) U^\dagger \lvert b \rangle\!\langle b \rvert U \mathcal D_{1, 3/f}^{\otimes k} (P_\mathbf q) U^\dagger \lvert b \rangle}_{\circled{1}}\!.\label{eq:shadow_norm_depolarizing_klocal_observable}
\end{align}

To evaluate $\circled{1}$, write $\sigma = \sum_{i_1\ldots i_k} \sigma_{i_1 \ldots i_k} e_{i_1} \otimes \ldots \otimes e_{i_k}$. Then, 

\begin{align}
    \circled{1} 
    &= \underset{U \sim \localclifford^{\otimes k}}{\E} \sum_{b \in \B^k} \langle b \rvert \mathcal D_{1,f}^{\otimes k}(U \sum_{i_1\ldots i_k} \sigma_{i_1 \ldots i_k} e_{i_1} \otimes \ldots \otimes e_{i_k}U^\dagger) \lvert b \rangle\nn &\qquad\cdot\langle b \rvert U \mathcal D_{1, 3/f}^{\otimes k} (P_\mathbf p) U^\dagger \lvert b \rangle\!\langle b \rvert U \mathcal D_{1, 3/f}^{\otimes k} (P_\mathbf q) U^\dagger \lvert b \rangle \nn
    &= \sum_{i_1 \ldots i_k}\!\!\sigma_{i_1 \ldots i_k} \bigotimes_{j=1}^k \underset{U_j \sim \localclifford}{\E} \sum_{b_j \in \B} \langle b_j \rvert \mathcal D_{1,f}(U_j e_{i_j} U_j^\dagger) \lvert b_j \rangle\nn&\qquad\cdot\langle b_j \rvert U_j \mathcal D_{1, 3/f} (P_{\mathbf p_j}) U_j^\dagger \lvert b_j \rangle\!\langle b_j \rvert U_j \mathcal D_{1, 3/f} (P_{\mathbf q_j}) U_j^\dagger \lvert b_j \rangle \nn
    &= \sum_{i_1 \ldots i_k}\!\!\sigma_{i_1 \ldots i_k} \prod_{j=1}^k \underset{U_j \sim \localclifford}{\E} \sum_{b_j \in \B} \langle b_j \rvert (f U_j e_{i_j} U_j^\dagger + (1-f)\tr( U_j e_{i_j} U_j^\dagger) \frac{\mathbb I}{2}) \lvert b_j \rangle\nn&\qquad\cdot\langle b_j \rvert U_j \mathcal D_{1, 3/f} (P_{\mathbf p_j}) U_j^\dagger \lvert b_j \rangle\!\langle b_j \rvert U_j \mathcal D_{1, 3/f} (P_{\mathbf q_j}) U_j^\dagger \lvert b_j \rangle \nn
    &= \sum_{i_1 \ldots i_k}\!\!\sigma_{i_1 \ldots i_k} \prod_{j=1}^k \underset{U_j \sim \localclifford}{\E} \sum_{b_j \in \B} \Big\{ f\langle b_j \rvert U_j e_{i_j} U_j^\dagger \lvert b_j \rangle + \frac{1-f}{2} \tr(e_{i_j})\Big\}\nn&\qquad\cdot\langle b_j \rvert U_j \mathcal D_{1, 3/f} (P_{\mathbf p_j}) U_j^\dagger \lvert b_j \rangle\!\langle b_j \rvert U_j \mathcal D_{1, 3/f} (P_{\mathbf q_j}) U_j^\dagger \lvert b_j \rangle \nn
    &= \sum_{i_1 \ldots i_k}\!\!\sigma_{i_1 \ldots i_k} \prod_{j=1}^k
    \bigg\{ f \underbrace{\underset{U_j \sim \localclifford}{\E} \sum_{b_j \in \B}\!\! \langle b_j \rvert U_j e_{i_j} U_j^\dagger \lvert b_j \rangle\!\langle b_j \rvert U_j \mathcal D_{1, 3/f} (P_{\mathbf p_j}) U_j^\dagger \lvert b_j \rangle\!\langle b_j \rvert U_j \mathcal D_{1, 3/f} (P_{\mathbf q_j}) U_j^\dagger \lvert b_j \rangle}_{\circled{2}}\nn&\qquad+
    \frac{1-f}{2} \tr(e_{i_j})\underbrace{\underset{U_j \sim \localclifford}{\E} \sum_{b_j \in \B} \langle b_j \rvert U_j \mathcal D_{1, 3/f} (P_{\mathbf p_j}) U_j^\dagger \lvert b_j \rangle\!\langle b_j \rvert U_j \mathcal D_{1, 3/f}(P_{\mathbf q_j}) U_j^\dagger \lvert b_j \rangle}_{\circled{3}} \bigg\}.
\end{align}

To evaluate $\circled{2}$ and $\circled{3}$, define $\xi[A](p,q)$ as 
\[
\xi[A](p, q) \overset{\mathrm{def}}{=} \underset{U \sim \localclifford}{\E} \sum_{b \in \B} \langle b \rvert U A U^\dagger \lvert b \rangle\!\langle b \rvert U \mathcal D_{1, 3/f} (P_p) U^\dagger \lvert b \rangle\!\langle b \rvert U \mathcal D_{1, 3/f} (P_q) U^\dagger \lvert b \rangle.
\]
Hence, $\circled{2} = \xi[e_{i_j}](\mathbf{p}_j, \mathbf{q}_j)$ and $\circled{3} = \xi[\mathbb I](\mathbf{p}_j, \mathbf{q}_j)$. We will now find an expression for $\xi[A](p, q)$.
\vspace{0.15in}

\textit{Case 1:} $p= q= 0$.
\vspace{-0.10in}
\begin{align}
    \xi[A](0,0) 
    &= \underset{U \sim \localclifford}{\E} \sum_{b \in \B} \langle b \rvert U A U^\dagger \lvert b \rangle\!\langle b \rvert U \mathcal D_{1, 3/f} (\mathbb I) U^\dagger \lvert b \rangle\!\langle b \rvert U \mathcal D_{1, 3/f} (\mathbb I) U^\dagger \lvert b \rangle \nn
    &= \underset{U \sim \localclifford}{\E} \sum_{b \in \B} \langle b \rvert U A U^\dagger \lvert b \rangle \nn
    &= \tr(A).
\end{align}

\textit{Case 2:} $p\neq 0, q= 0$.
\vspace{-0.10in}
\begin{align}
    \xi[A](p,0) 
    &= \underset{U \sim \localclifford}{\E} \sum_{b \in \B} \langle b \rvert U A U^\dagger \lvert b \rangle\!\langle b \rvert U \mathcal D_{1, 3/f} (P_p) U^\dagger \lvert b \rangle\!\langle b \rvert U \mathcal D_{1, 3/f} (\mathbb I) U^\dagger \lvert b \rangle \nn
    &= \frac{3}{f}\underset{U \sim \localclifford}{\E} \sum_{b \in \B} \langle b \rvert U A U^\dagger \lvert b \rangle\!\langle b \rvert U P_p U^\dagger \lvert b \rangle \nn
    &= \frac{3}{f}\tr \Big\{ A \underset{U \sim \localclifford}{\E} \sum_{b \in \B} U^\dagger\lvert b\rangle\! \langle b \rvert U \langle b \rvert U P_p U^\dagger \lvert b \rangle \Big\} \nn
    &= \frac{3}{f}\tr \Big\{ A \mathcal{M}_{\localclifford, \mathbb I}(P_p) \Big\} \nn
    &= \frac{1}{f}\tr(A P_p). 
\end{align}

\textit{Case 3:} $p = 0, q\neq 0$.
By symmetry, 
\vspace{-0.10in}
\begin{align}
    \xi[A](0,q) 
    &= \frac{1}{f}\tr(A P_q). 
\end{align}

\textit{Case 4:} $p \neq 0, q \neq 0$.
\vspace{-0.10in}
\begin{align}
    \xi[A](p,q) 
    &= \underset{U \sim \localclifford}{\E} \sum_{b \in \B} \langle b \rvert U A U^\dagger \lvert b \rangle\!\langle b \rvert U \mathcal D_{1, 3/f} (P_p) U^\dagger \lvert b \rangle\!\langle b \rvert U \mathcal D_{1, 3/f} (P_q) U^\dagger \lvert b \rangle \nn
    &=\frac{9}{f^2} \tr \Big\{ A \underset{U \sim \localclifford}{\E} \sum_{b \in \B} U^\dagger\lvert b\rangle\! \langle b \rvert U \langle b \rvert U P_p U^\dagger \lvert b \rangle\!\langle b \rvert U \mathcal P_q U^\dagger \lvert b \rangle \Big\} \nn
    &=\frac{9}{f^2} \tr \Big\{ A \frac{1}{3} \delta_{pq} \mathbb I \Big\} \nn
    &=\frac{3}{f^2} \delta_{pq} \tr(A).
\end{align}
The third equality follows from Eq.~(S36) of \cite{huang2020predicting}, which itself follows from \cref{eq:generalizedEqS36} by setting $\mathcal E = I$, $d=2$, and $\tr(B)= \tr(C)=0$. Combining the four cases gives 

\begin{align}
    \xi[A](p,q) = \begin{cases}
    \tr(A) &\text{if $p = q = 0$.} \\
    \frac{1}{f}\tr(AP_p) &\text{if $p \neq 0, q = 0$.} \\
    \frac{1}{f}\tr(AP_q) &\text{if $p = 0, q \neq 0$.} \\\
    \frac{3}{f^2}\tr(A) &\text{if $p = q \neq 0$.} \\
    0 &\text{if $p \neq q, p\neq 0, q \neq 0$.} \\
    \end{cases} = \zeta(p,q)\tr(AP_pP_q), 
\end{align}
where 
\begin{align}
    \zeta(p,q) = \begin{cases}
    1/f &\text{if $p=0$ or $q=0$.}\\
    3/f^2 &\text{if $p=q \neq 0$.}\\
    0 &\text{otherwise.}\\
    \end{cases}
\end{align}
Applying this to $\circled{2}$ and $\circled{3}$ gives 
\[
\circled{2} = \xi[e_{i_j}](\mathbf{p}_j, \mathbf{q}_j) = 
\zeta(\mathbf{p}_j, \mathbf{q}_j)\tr(e_{i_j}P_{\mathbf{p}_j}P_{\mathbf{q}_j})
\]
and 
\[
\circled{3} = \xi[\mathbb I](\mathbf{p}_j, \mathbf{q}_j) = \zeta(\mathbf{p}_j, \mathbf{q}_j)\tr(P_{\mathbf{p}_j}P_{\mathbf{q}_j}) = 
2\zeta(\mathbf{p}_j, \mathbf{q}_j)\delta_{\mathbf p_j \mathbf q_j}.
\]
Plugging these expressions into $\circled{1}$ gives 
\begin{align}
\circled{1}
&= \sum_{i_1 \ldots i_k}\!\!\sigma_{i_1 \ldots i_k} \prod_{j=1}^k\bigg\{ f \zeta(\mathbf{p}_j, \mathbf{q}_j)\tr(e_{i_j}P_{\mathbf{p}_j}P_{\mathbf{q}_j})    +\frac{1-f}{2} \tr(e_{i_j})2\zeta(\mathbf{p}_j, \mathbf{q}_j)\delta_{\mathbf p_j \mathbf q_j}\bigg\} \nn
&= \sum_{i_1 \ldots i_k}\!\!\sigma_{i_1 \ldots i_k}\Big(\prod_{j=1}^k \zeta(\mathbf{p}_j, \mathbf{q}_j)\Big) \prod_{j=1}^k\bigg\{ \underbrace{f \tr(e_{i_j}P_{\mathbf{p}_j}P_{\mathbf{q}_j})    +(1-f) \tr(e_{i_j})\delta_{\mathbf p_j \mathbf q_j}}_{\circled{4}}\bigg\}.
\end{align}
When $\mathbf p_j \neq \mathbf q_j$, $\circled{4} = f \tr(e_{i_j}P_{\mathbf p_j}P_{\mathbf q_j})$. When $\mathbf p_j = \mathbf q_j$, $\circled{4} = f \tr(e_{i_j} P_{\mathbf p_j}^2) + (1-f)\tr(e_{i_j}) = \tr(e_{i_j})$. Hence, 
\[\circled{4} = f^{\mathbbm 1_{\mathbf p_j \neq \mathbf q_j}} \tr(e_{i_j}P_{\mathbf p_j}P_{\mathbf q_j}).\] 
\begin{align}
 \circled{1}
&= \sum_{i_1 \ldots i_k}\!\!\sigma_{i_1 \ldots i_k}\Big(\prod_{j=1}^k \zeta(\mathbf{p}_j, \mathbf{q}_j)\Big) \prod_{j=1}^k\bigg\{ f^{\mathbbm 1_{\mathbf p_j \neq \mathbf q_j}} \tr(e_{i_j}P_{\mathbf p_j}P_{\mathbf q_j})\bigg\} \nn
&= \prod_{j=1}^k \underbrace{f^{\mathbbm 1_{\mathbf p_j \neq \mathbf q_j}}\zeta(\mathbf{p}_j, \mathbf{q}_j)}_{\circled{5}} \sum_{i_1 \ldots i_k}\!\!\sigma_{i_1 \ldots i_k} \tr((e_{i_1} \otimes \ldots \otimes e_{i_k})(P_{\mathbf p_1} \otimes \ldots \otimes P_{\mathbf p_k})(P_{\mathbf q_1} \otimes \ldots \otimes P_{\mathbf q_k})).
\end{align}

Let $\tilde{f}(p, q) = f^{\mathbbm 1_{\mathbf p_j \neq \mathbf q_j}}\zeta(\mathbf{p}_j, \mathbf{q}_j)$. Then, 
\begin{align}
    \tilde{f}(p, q) = \begin{cases}
    f &\text{if $p \neq q$.} \\
    1 &\text{if $p = q$.} \\
    \end{cases} \times 
    \begin{cases}
    1/f &\text{if $p=0$ or $q=0$.}\\
    3/f^2 &\text{if $p=q \neq 0$.}\\
    0 &\text{otherwise.}\\
    \end{cases} = 
    \begin{cases}
    1/f &\text{if $p=q=0$.} \\
    1 &\text{if $(p=0) \oplus (q=0)$.} \\
    3/f^2 &\text{if $p=q\neq 0$.} \\
    0 &\text{otherwise.}
    \end{cases}
\end{align}
and $\circled{5} = \tilde{f}(\mathbf p_j, \mathbf q_j)$. Let $\tilde{\mathcal F}(\mathbf p, \mathbf q) = \prod_{j= 1}^k \tilde{f}(\mathbf p_j, \mathbf q_j)$. Then,  

\begin{align}
 \circled{1}
&= \tilde{\mathcal{F}}(\mathbf p, \mathbf q) \sum_{i_1 \ldots i_k}\!\!\sigma_{i_1 \ldots i_k} \tr((e_{i_1} \otimes \ldots \otimes e_{i_k})P_{\mathbf p}P_{\mathbf q}) \nn
&= \tilde{\mathcal{F}}(\mathbf p, \mathbf q) \tr(\sigma P_{\mathbf p}P_{\mathbf q}).
\end{align}

Plugging the expression for $\circled{1}$ into \cref{eq:shadow_norm_depolarizing_klocal_observable} completes the proof. 

\begin{align}
    \norm{O}^2_{\shadow, \localclifford^{\otimes n}, \mathcal D_{1,f}^{\otimes n}}
    &= \max_{\sigma \in \mathbb D_{2^k}}\! \sum_{\mathbf p, \mathbf q \in \mathbb Z_4^k}\!\!\! \alpha_\mathbf p \alpha_\mathbf q \tilde{\mathcal{F}}(\mathbf p, \mathbf q) \tr(\sigma P_{\mathbf p}P_{\mathbf q}) \nn
    &= \max_{\sigma \in \mathbb D_{2^k}} \tr\Big\{ \sigma  \sum_{\mathbf p, \mathbf q \in \mathbb Z_4^k}\!\!\! \alpha_\mathbf p \alpha_\mathbf q \tilde{\mathcal{F}}(\mathbf p, \mathbf q) P_{\mathbf p}P_{\mathbf q}\Big\} \nn
    &= \norm{\sum_{\mathbf p, \mathbf q \in \mathbb Z_4^k}\!\!\! \alpha_\mathbf p \alpha_\mathbf q \tilde{\mathcal{F}}(\mathbf p, \mathbf q) P_{\mathbf p}P_{\mathbf q}}_{\mathrm{sp}}\!\!.
\end{align}

\end{proof}

\section{Noisy Input States}
\label{sec:noisy_input_states}

Throughout the main text, our assumption has been that the input state $\rho$ is prepared without any errors. But what if the input state given is itself noisy? In this appendix, we consider the case where
in addition to the noise described in the noisy measurement primitive of \cref{def:noisy_measurement_primitive}, the input state $\rho$ is subject to the noise channel $\mathcal K$. In other words, instead of the intended transformation $\rho \mapsto U\rho U^\dag$, the input state transforms as $\rho \mapsto \mathcal E(U \mathcal K(\rho) U^\dag)$. This scenario is equivalent to a noise model where the unitary operation $U$ is replaced by one where a noise channel acts both before and after the perfect implementation of $U$ (see \cref{footnote:GTM_noise}).

With this change,  \cref{eq:output_procedure} becomes
\begin{align}
\label{eq:output_procedure_new}
U^\dagger \lvert \hat{b} \rangle\!\langle \hat{b} \rvert U \quad \text{with probability} \quad P_b(\hat{b}) \overset{\text{def}}{=} \langle \hat{b}\rvert \channel(U \mathcal K(\rho) U^\dagger) \lvert \hat{b} \rangle \quad \text{where} \quad \hat{b} \in \{0,1\}^n.
\end{align}
The noisy shadow channel in \cref{eq:shadow_channel} is modified to
    \begin{align}\label{eq:shadow_channel_new}
    \mathcal M_{\ensemble,\channel, \mathcal K}(\rho) \overset{\mathrm{def}}{=}  \E_{U \sim \ensemble} \sum_{b \in \{0,1\}^n} \bra{b} \channel(U \mathcal K(\rho) U^\dagger) \ket{b} U^\dagger \ket{b}\!\bra{b} U = 
    \left(\shadowchannel{\ensemble}{\channel}\circ \mathcal K\right)(\rho).
    \end{align}
and the noisy classical shadow of \cref{eq:classical_shadow} becomes
\begin{align}\label{eq:classical_shadow_new}
\hat{\rho} &= \hat{\rho}(\ensemble, \channel, \mathcal K, \hat{U}, \hat{b}) \overset{\mathrm{def}}{=} 
\mathcal M^{-1}_{\ensemble,\channel,\mathcal K}(\hat{U}^\dagger \lvert \hat{b}\rangle\!\langle \hat{b} \rvert \hat{U})
\nonumber\\
&= (\mathcal K^{-1} \circ \mathcal M_{\ensemble,\channel}^{-1})(\hat{U}^\dagger \lvert \hat{b}\rangle\!\langle \hat{b} \rvert \hat{U})
,
\end{align}
where we have assumed that both the shadow channel and the noise channel $\mathcal K$ are invertible linear superoperators. As before, we do not assume that the inverses are themselves quantum channels.

Next, the shadow seminorm of \cref{eq:noisy_shadow_norm} becomes modified to
\begin{align}
\label{eq:noisy_shadow_norm_new}
\norm{O}_{\shadow, \ensemble, \channel,\mathcal K} &= \max_{\sigma \in \mathbb{D}_{2^n}} 
\sqrt{\underset{U \sim \ensemble}{\E} \sum_{b \in \{0,1\}^n} \langle b \rvert \channel(U \sigma U^\dagger) \lvert b \rangle\!\langle b \rvert U \mathcal{M}^{-1,\dagger}_{\ensemble, \channel,\mathcal K}(O) U^\dagger \lvert b \rangle^2 }
\nonumber \\
&= \max_{\sigma \in \mathbb{D}_{2^n}} 
\sqrt{\underset{U \sim \ensemble}{\E} \sum_{b \in \{0,1\}^n} \langle b \rvert \channel(U \sigma U^\dagger) \lvert b \rangle\!\langle b \rvert U
\mathcal K^{-1,\dagger}
(\mathcal{M}^{-1,\dagger}_{\ensemble, \channel}(O)) U^\dagger \lvert b \rangle^2 }.
\end{align}

Hence, by following the same argument as in the main text, we arrive at \cref{thm:general-theorem}, but with the shadow seminorm $\norm{\cdot}_{\shadow, \ensemble, \channel}$ replaced by
$\norm{\cdot}_{\shadow, \ensemble, \channel,\mathcal K}$. Consequently, the only changes needed to adapt \cref{algo:full_protocol} to this case are
\begin{enumerate}
    \item In step \ref{step:full_protocol_2} of \cref{algo:full_protocol}, replace the shadow seminorm with \cref{eq:noisy_shadow_norm_new}.
    \item In step \ref{step:full_protocol_5} of \cref{algo:full_protocol}, replace $\mathcal E(U\rho U^\dag)$ with  $\mathcal E(U \mathcal K(\rho) U^\dag)$.
    \item In step \ref{step:full_protocol_8} of \cref{algo:full_protocol}, replace each occurrence of $\mathcal M^{-1}_{\mathcal U,\mathcal E}$ with $\mathcal K^{-1} \circ \mathcal M^{-1}_{\mathcal U,\mathcal E}$.
\end{enumerate}

\bibliographystyle{unsrt}
\bibliography{refs}

\end{document}